\RequirePackage{afterpackage}
\AfterPackage{amsthm}{
  \RequirePackage{hyperref}
}

\documentclass[english,a4paper,cleveref]{lipics-v2021}

\nolinenumbers

\usepackage[utf8]{inputenc}

\newcommand{\sfset}{\textsf{set}}
\newcommand{\stset}[1]{\st{#1}_{\sfset}}

\usepackage{silence} 
\usepackage{paralist}
\usepackage{amsmath} 
\usepackage{fix-cm}

\usepackage{relsize} 
\usepackage{stmaryrd} 
\usepackage{amssymb} 
\usepackage{mathtools} 
\usepackage{prooftree} 
\usepackage{datetime2} 

\usepackage[dvipsnames,svgnames]{xcolor}

\usepackage{bm}

\usepackage{enumerate} 
\usepackage{graphics}
\usepackage{xspace}
\usepackage{centernot}
\usepackage{cancel}
\usepackage{graphicx}
\usepackage{dsfont}
\usepackage[normalem]{ulem} 
\usepackage{mdframed}
\usepackage{yfonts} 
\usepackage[bb=boondox]{mathalfa}
\usepackage[normalem]{ulem} 

\mdfsetup{skipabove=1em,skipbelow=1em}

\usepackage{ wasysym }
\usepackage{ upgreek }

\newcommand{\ctxleq}{\ensuremath{\sqsubseteq_{\mathit{ctx}}}}

\newcommand{\sqsubsetsim}{\vcenter{\offinterlineskip\hbox{$\sqsubset$}\vskip 0.2ex\hbox{$\sim$}}}

\newcommand{\Nsqsubsetsim}{\centernot{\vcenter{\offinterlineskip\hbox{$\sqsubset$}\vskip 0.2ex\hbox{$\sim$}}}}

\newcommand{\testleq}[1][\envTEST]{\mathrel{\sqsubsetsim^{\mbox{\scalebox{.6}{\ensuremath{#1}}}}}}
\newcommand{\Ntestleq}[1][\envTEST]{\ensuremath{\mathrel{\Nsqsubsetsim^{^{\kern-3pt#1}}}}}

\newcommand{\testleqset}[1][\envTEST]{\mathrel{\sqsubsetsim_\sfset^{\mbox{\scalebox{.6}{\ensuremath{#1}}}}}\kern-3pt}

\newcommand{\testeq}[1][]{\ensuremath{\mathrel{\eqsim^{#1}}}}

\newcommand{\altleq}[1][]{\ensuremath{\mathrel{\leq^{#1}_{\mathit{alt}}}}}
\newcommand{\Naltleq}[1][]{\ensuremath{\mathrel{\nleq^{#1}_{\mathit{alt}}}}}
\newcommand{\extleq}{\ensuremath{\mathrel{\leq_{\mathit{ext}}}}}
\newcommand{\altleqset}[1][\labs]{\ensuremath{\mathrel{\leq^{#1}_\sfset}}}

\newcommand{\wsSym}[1][]{\mathit{WS}^{#1}}
\newcommand{\ws}[3]{\wsSym[#1](#2,#3)}

\newcommand{\sem}[1]{\llbracket #1 \rrbracket}

\newcommand{\interpSym}[2][\labs]{\Lbag - \Rbag^{#2}_{#1}}
\newcommand{\interp}[3][\labs]{\Lbag #3 \Rbag^{#2}_{#1}}

\newcommand{\agents}[1]{\ensuremath{\textsf{L}^{#1}}}
\newcommand{\ltsmultiset}[1]{\ensuremath{\textsf{FW}^{\, #1}}}

\newcommand{\conv}{\downarrow}
\newcommand{\cnvalong}{\ensuremath{\mathrel{\conv}}}

\newcommand{\fw}{\mathsf{fw}}
\newcommand{\stfw}[1]{\ensuremath{\mathrel{\overset{#1} \longrightarrow_\fw}}}
\newcommand{\wta}[1]{\ensuremath{\mathrel{\overset{#1} \Longrightarrow_\fw}}}

\newcommand{\Names}{\text{\tt Names}\xspace}

\newcommand{\bla}{\beta}  
\newcommand{\blb}{\beta'} 
\newcommand{\blc}{\beta''} 

\renewcommand{\aa}{\mu} 
\newcommand{\ab}{\mu'} 
\newcommand{\ac}{\mu''} 

\newcommand{\nba}{\eta}  
\newcommand{\nbb}{\eta'} 
\newcommand{\nbc}{\eta''} 

\newcommand{\co}[1]{\overline{#1}}

\newcommand\Act{{\tt Act}\xspace}
\newcommand\Actfin{\ensuremath{\Act^\star}\xspace}

\newcommand\Effects{\Act}

\newcommand\Out{\text{{\tt Out}}\xspace}

\newcommand\ActV{\text{{\tt Actv}}\xspace}

\newcommand\OutV{\text{{\tt OutV}}\xspace}

\newcommand\Labels{\Acttau}
\newcommand\BLabels{\text{{\tt B}}\xspace}

\newcommand\NBLabels{\text{{\tt N}}\xspace}

\newcommand\Effectfin{\ensuremath{\Effects^\star}\xspace}

\newcommand\Acttau{\ensuremath{{\Effects_\tau}}\xspace}

\newcommand{\mathsc}[1]{\textup{\textsc{#1}}}

\newcommand{\rulename}[1]{{\mathsc{[#1]}}}
\newcommand{\rname}[1]{\rulename{#1}}

\newcommand{\parL}{\rname{Par-L}\xspace}
\newcommand{\parR}{\rname{Par-R}\xspace}
\newcommand{\com}{\rname{Com}\xspace}

\newcommand{\rinput}{\rname{Input}}
\newcommand{\routput}{\rname{Output}}

\newcommand{\unfold}{\rname{Unf}}
\newcommand{\rres}{\rname{Res}}
\newcommand{\rthen}{\rname{Then}}
\newcommand{\relse}{\rname{Else}}
\newcommand{\extL}{\rname{Sum-L}}
\newcommand{\extR}{\rname{Sum-R}}

\newcommand{\scom}{\rname{S-com}}
\newcommand{\stauserver}{\rname{S-Srv}}
\newcommand{\stauclient}{\rname{S-Clt}}

\newcommand{\axiom}[1]{\textup{\textsc{#1}}\xspace}

\newcommand{\nonblocking}{nb}
\newcommand{\nbdelay}{\axiom{\hyperlink{nb-delay}{\nonblocking-delay}}}
\newcommand{\nbconfluence}{\axiom{\hyperlink{nb-confluence}{\nonblocking-confluence}}}
\newcommand{\nbdeterminacy}{\axiom{\hyperlink{nb-determinacy}{\nonblocking-determinacy}}}
\newcommand{\nbfeedback}{\axiom{\hyperlink{nb-feedback}{feedback}}}

\newcommand{\nbtau}{\axiom{\hyperlink{nb-tau}{\nonblocking-tau}}}
\newcommand{\nbdeterminacyinv}{\axiom{\hyperlink{nb-inv-determinacy}{backwards-\nonblocking-determinacy}}}

\newcommand{\fwdfeedback}{\axiom{\hyperlink{fwd-feedback}{Fwd-Feedback}}}
\newcommand{\boom}{\axiom{\hyperlink{boomerang}{Boomerang}}}

\newcommand{\enabled}{\axiom{cn-enabled}}

\newcommand\set[1]{\{ #1 \}}
\newcommand\setof[2]{\{ #1 \ |\  #2 \}}

\newcommand\parts[1]{\mathcal{P}(#1)}
\newcommand\pparts[1]{\mathcal{P}^{+}_{fin}(#1)}
\newcommand\fparts[1]{\mathcal{P}_{fin}(#1)}

\newcommand{\mset}[1]{\{\!|#1|\!\}}

\newcommand{\cardinality}[1]{\mid #1 \mid}

\newcommand{\modulo}[2]{{#1}_{#2}}

\def\BaseUrl{https://glmaths.github.io/testing-theory/TestingTheory.} 
\newcommand{\coqlink}[1]{\href{\BaseUrl#1}{\raisebox{-0.6mm}{\includegraphics[height=0.8em] {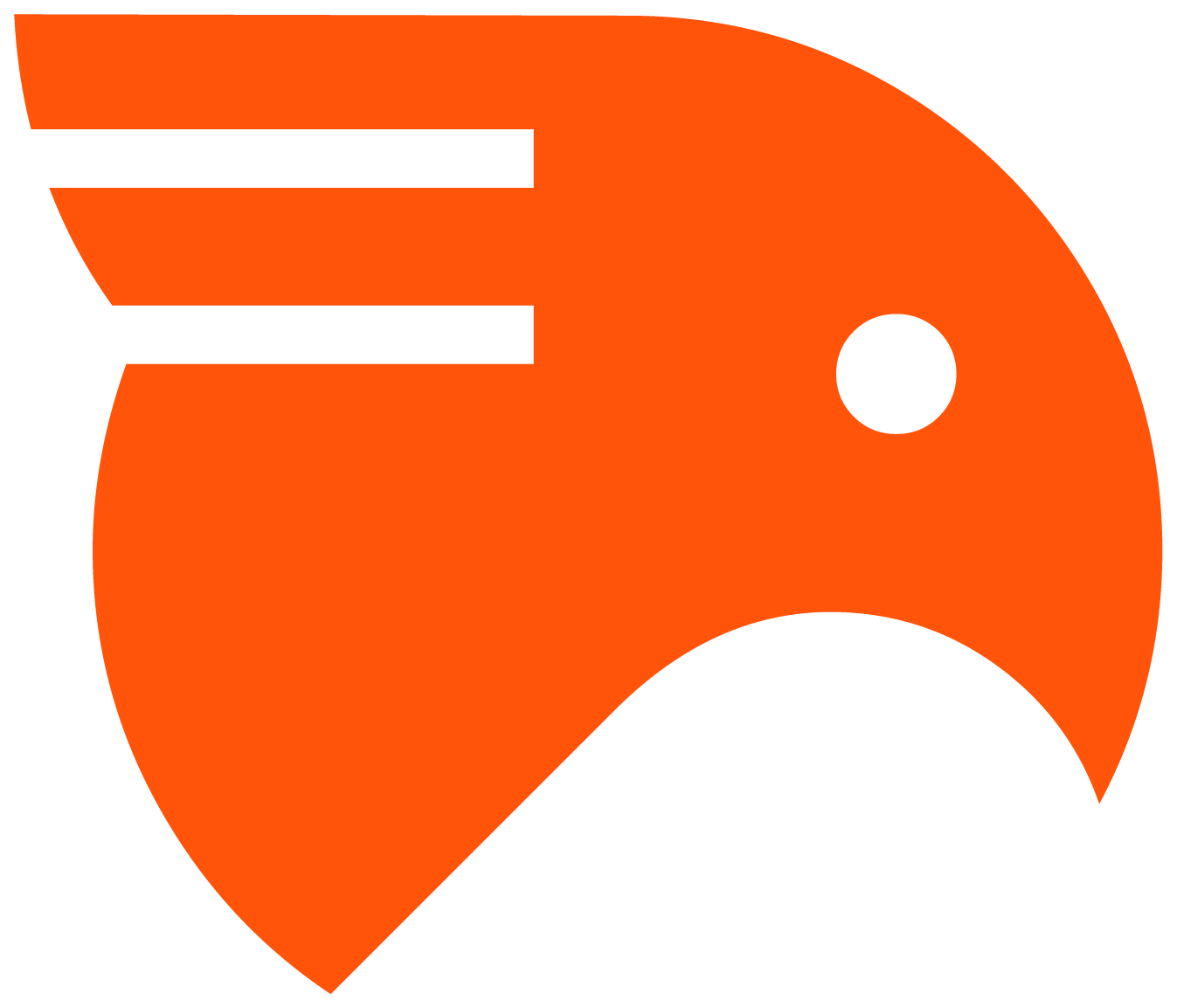}}}}

\newcommand{\myshow}[1]{#1}

\newcommand\coqVACCSEx[1]{\myshow{\coqlink{Labelled_Transition_System.Examples_of_LTS.VACCS.VACCS_Examples.html\##1}}}
\newcommand\coqToSet[1]{\myshow{\coqlink{Interaction_Between_Lts.Examples_of_Interactions.ParallelLTSConstruction.html\##1}}}
\newcommand\coqMustE[1]{\myshow{\coqlink{Testing_Pre_orders.Must.MustE.html\##1}}}
\newcommand\coqAbs[1]{\myshow{\coqlink{Testing_Pre_orders.Must.Alternative_Preorder.Acceptance_Set.DefinitionAS.html\##1}}}
\newcommand\coqFiniteImage[1]{\myshow{\coqlink{Labelled_Transition_System.Definitions.FiniteImageLTS.html\##1}}}
\newcommand\coqEquivAS[1]{\myshow{\coqlink{Testing_Pre_orders.Must.Alternative_Preorder.Acceptance_Set.Equivalence.html\##1}}}
\newcommand\coqToFW[1]{\myshow{\coqlink{Interaction_Between_Lts.Examples_of_Interactions.ForwarderConstruction.html\##1}}}
\newcommand\coqgLTSOBA[1]{\myshow{\coqlink{Labelled_Transition_System.Definitions.Lts_OBA.html\##1}}}
\newcommand\coqgLTSFW[1]{\myshow{\coqlink{Labelled_Transition_System.Definitions.Lts_FW.html\##1}}}
\newcommand\coqLift[1]{\myshow{\coqlink{Testing_Pre_orders.Must.Lift.html\##1}}}
\newcommand\coqSoundAS[1]{\myshow{\coqlink{Testing_Pre_orders.Must.Alternative_Preorder.Acceptance_Set.Soundness.html\##1}}}
\newcommand\coqCpltAS[1]{\myshow{\coqlink{Testing_Pre_orders.Must.Alternative_Preorder.Acceptance_Set.Completeness.html\##1}}}

\newcommand\coqMutiSetLTS[1]{\myshow{\coqlink{Labelled_Transition_System.Examples_of_LTS.MULTISET.MultisetLTSConstruction.html\##1}}}
\newcommand\coqSetLTS[1]{\myshow{\coqlink{Labelled_Transition_System.Examples_of_LTS.SET_LTS.SetLTSConstruction.html\##1}}}
\newcommand\coqACCSCorolarry[1]{\myshow{\coqlink{Labelled_Transition_System.Examples_of_LTS.ACCS.ACCS_Must_Characterization.html\##1}}}
\newcommand\coqCCSCorolarry[1]{\myshow{\coqlink{Labelled_Transition_System.Examples_of_LTS.CCS.CCS_Must_Characterization.html\##1}}}
\newcommand\coqVACCSCorolarry[1]{\myshow{\coqlink{Labelled_Transition_System.Examples_of_LTS.VACCS.VACCS_Must_Characterization.html\##1}}}
\newcommand\coqVCCSCorolarry[1]{\myshow{\coqlink{Labelled_Transition_System.Examples_of_LTS.VCCS.VCCS_Must_Characterization.html\##1}}}

\DeclareMathOperator{\opMust}{\mathsc{must}}
\DeclareMathOperator{\opMay}{\mathsc{may}}

\newcommand{\opMusti}{\opMust\ensuremath{_i}}
\newcommand{\musti}[2]{\ensuremath{#1 \opMusti #2}}
\newcommand{\Nmusti}[2]{\ensuremath{#1 \centernot{\opMusti} #2}}

\newcommand{\Must}[2]{\ensuremath{#1 \opMust #2}}
\newcommand{\NMust}[2]{\ensuremath{#1 \centernot{\opMust} #2}}

\newcommand{\opMustset}{\opMust_{\sfset}}

\newcommand{\mustset}[2]{\ensuremath{#1 \opMustset #2}}

\renewcommand{\and}{\text{ and }}
\renewcommand{\lor}{\text{ or }}
\renewcommand{\implies}{\text{ implies }}
\newcommand{\imply}{\text{ imply }}

\newcommand{\forevery}{\text{for every }}
\newcommand{\Forevery}{\text{For every }}

\newcommand{\wehavethat}{. \;}
\newcommand{\suchthat}{\wehavethat}

\newcommand{\trace}{s}

\newcommand{\state}[1][0]{p_{#1}}
\newcommand{\stateA}{\state[1]}
\newcommand{\stateB}{\state[2]}
\newcommand{\stateC}{\state[3]}
\newcommand{\stateD}{\state[4]}

\newcommand{\N}[1][]{\mathbb{N}_{#1}}

\newcommand{\dom}[1]{\mathit{dom}(#1)}

\newcounter{thm}
\newcommand{\thistheoremname}{}
\newtheorem{genericlem}[thm]{\thistheoremname}

\newcommand{\accSym}{\mathcal{A}}
\newcommand{\acc}[2]{\accSym(#1,#2)}

\newcommand{\leaveout}[1]{}

\newcommand{\eqdef}{\mathrel{\stackrel{\mathsf{def}}{=}}}

\makeatletter
\newcommand{\subsetsim}{\mathrel{\mathpalette\subset@sim\relax}}
\newcommand{\subset@sim}[2]{%
  \vtop{\offinterlineskip\m@th
    \ialign{\hfil##\cr
      $#1\sqsubset$\cr\noalign{\kern0.5pt}\scalebox{0.9}{$#1\sim$}\cr
   }%
  }%
}
\makeatother

\newcommand{\mysmash}[1]{\smash{#1}} 
\renewcommand{\st}[1]{\ensuremath{\mathrel{\mysmash{\xrightarrow{#1}}}}}

\newcommand{\Nst}[1]{\ensuremath{\mathrel{\mysmash{\overset{#1}{\longarrownot\longrightarrow}}}}}
\newcommand{\stable}{\Nst{\intaction}}

\newcommand{\wt}[1]{\ensuremath{\mathrel{\mysmash{\overset{#1} \Longrightarrow}}}}

\newcommand\StatesA{A}
\newcommand\StatesB{B}

\newcommand\States{S}

\newcommand{\subst}[2]{[^{#1}/_{#2}]}

\newcommand{\testconvSym}{\mathit{tc}} 
\newcommand{\testconv}[1]{\testconvSym(#1)} 

\newcommand{\testaccSym}{\mathit{ta}} 
\newcommand{\testacc}[2]{\testaccSym(#1,#2)} 

\newcommand{\LTS}[3]{\ensuremath{\langle#1,#2,#3\rangle}}
\newcommand{\lts}{\mathcal{L}}

\newcommand{\LTSPar}{\Par}

\newcommand{\client}{t}
\newcommand{\server}{p}
\newcommand{\serverA}{p}
\newcommand{\serverB}{q}
\newcommand{\serverC}{q'}
\newcommand{\serverD}{p'}

\newcommand{\goodSym}{\mathsc{good}}
\newcommand{\good}[1]{\goodSym(#1)}

\newcommand{\MO}{\mathit{MO}}

\newcommand{\axcmpl}{\mathsc{AxCmpl}}

\newcommand{\liftFWSym}{\mathsf{toFW}}
\newcommand{\liftFW}[1]{\liftFWSym(#1)}

\newcommand{\toSetSym}{\mathsf{toSET}}
\newcommand{\toSet}[1]{\toSetSym(#1)}

\newcommand{\wrt}{w.r.t.\xspace}
\newcommand{\LTSs}{LTSs\xspace}

\newcommand{\mustpreorder}{$\opMust$-preorder\xspace}

\newcommand{\ActAsTau}[0]{\tau}

\newcommand{\blockingaction}{\bla}
\newcommand{\nonblockingaction}{\nba}
\newcommand{\extaction}{\aa}
\newcommand{\extactionA}{\extaction}
\newcommand{\extactionB}{\extaction'}

\newcommand{\anyaction}{\alpha}
\newcommand{\anyactionb}{\alpha'}

\newcommand{\someMSET}{M}
\newcommand{\coRsym}{\mathsf{coR}}
\newcommand{\coR}[1]{\ensuremath{\coRsym(#1)}}
\newcommand{\coRlatest}[1]{\latest{\coR{#1}}}
\newcommand{\coRlatestSym}{\latestSym \circ \coRsym}

\newcommand{\readysetSym}{\mathsf{R}}
\newcommand{\readyset}[1]{\ensuremath{\readysetSym(#1)}} 

\newcommand{\unionofpreactionspars}[3]{\left( \bigcup\limits_{{#1}' \in {\ws{\lts}{#1}{#2}}} {#3}({#1}') \right)}

\newcommand{\preactionsubset}{E}
\newcommand{\preactionsubsetelem}{e}
\newcommand{\envTEST}{\mathcal{T}}
\newcommand{\envServerA}{\StatesA}

\newcommand{\someserverset}{X}
\newcommand{\someserversetA}{X}
\newcommand{\someserversetB}{Y}
\newcommand{\somereadyset}{R}

\newcommand{\clientA}{\client}

\newcommand{\test}{t}
\newcommand{\testA}{t_1}
\newcommand{\testB}{t_2}

\newcommand{\idVACCS}[1][a]{\mathsf{ccat}(#1)} 
  \newcommand{\lcc}[1]{\InputV{#1}{x}{\OutputV{#1}{x}{\Nil}}}
\newcommand{\constVACCS}[1][a]{\mathsf{cst}(#1)} 
\newcommand{\constVACCSv}[2][a]{\mathsf{cst}_{#2}(#1)} 

\newcommand{\ctest}{t}

\newcommand{\inputc}[1]{{#1} ?}
\newcommand{\outputc}[1]{{#1} !}
\newcommand\CCS{\mathsc{CCS}\xspace}
\newcommand\ACCS{\mathsc{ACCS}\xspace}

\newcommand\VCCS{\mathsc{VCCS}\xspace}
\newcommand\VACCS{\mathsc{VACCS}\xspace}

\newcommand{\VRP}{X} 

\newcommand{\ungoodSym}{\neg \mathsc{good}}
\newcommand{\ungood}[1]{\ungoodSym \left ( {#1}  \right )}
\newcommand{\emptytrace}{\varepsilon}

\newcommand{\stLeft}{\rname{S-Left}}
\newcommand{\stRight}{\rname{S-Right}}
\newcommand{\stInter}{\rname{S-Inter}}

\newcommand{\ActTau}[0]{\tau}

\newcommand{\ActInV}[2]{{#1}?{#2}}   
\newcommand{\ActOutV}[2]{{#1}!{#2}} 

\newcommand{\somechannel}{a}
\newcommand{\somechannelb}{b}

\newcommand{\be}{be}
\newcommand{\True}{True}
\newcommand{\False}{False}

\newcommand{\somevalue}{v}

\newcommand{\ActInOn}[1]{{#1}?}
\newcommand{\ActOutOn}[1]{{#1}!}

\newcommand{\FinA}{Y}
\newcommand{\deltadualSym}{\delta}

\newcommand{\deltadualphifiniteSym}{(\deltadualSym \circ \latestSym)}
\newcommand{\deltadualphifinite}[1]{\deltadualphifiniteSym ({#1})}

\newcommand{\labs}{\mathbb{A}}

\newcommand{\latestSym}{\phi}
\newcommand{\latest}[1]{\latestSym(#1)}
\newcommand{\laprogSym}{\delta}
\newcommand{\laprog}[1]{\laprogSym(#1)} 

\newcommand{\BNFsep}{\;\;|\;\;}

\newcommand{\Nil}{\mathbb{0}}
\newcommand{\Unit}{\mathbb{1}}
\newcommand{\mailbox}[1]{\ensuremath{#1}}

\newcommand{\extc}{\mathrel{+}}

\newcommand{\Par}{\mathrel{\parallel}}

\newcommand{\rec}[2][\VRP]{\text{\tt rec} #1. #2}

\newcommand{\csys}[2]{#1 \mathrel{\llceil} #2}

\newcommand{\action}{\alpha}

\newcommand{\domOfTest}{T}
\newcommand{\domOfServer}{P}
\newcommand{\domOfServerA}{P}
\newcommand{\domOfServerB}{Q}

\newcommand{\Inputon}[1]{{#1} ?}
\newcommand{\Outputon}[1]{{#1} !}

\newcommand{\map}[3]{{#1}_{#2} (#3)}
\newcommand{\mapSym}[2]{{#1}_{#2}}
\newcommand{\Val}{\text{{\tt Val}}}
\newcommand{\unit}{\text{{\tt ()}}}

\newcommand{\PreA}{X}

\newcommand{\badset}{ X_\textsf{bad}}

\newcommand{\ParB}{\mathrel{\parallel}}
\newcommand{\Prl}[2]{#1 \ParB #2}
\newcommand{\Choice}[2]{{#1} \extc {#2}}

\newcommand{\If}{\text{\tt if}}
\newcommand{\Then}{\text{\tt then}}
\newcommand{\Else}{\text{\tt else}}
\newcommand{\resC}{\nu}
\newcommand{\ResChan}[2]{(\resC {#1}){#2}}
\newcommand{\IfTE}[3]{\If \: {#1} \: \Then \: {#2} \: \Else \: {#3}}

\newcommand{\Success}{\Unit}
\newcommand{\InputV}[3]{\ActInV{#1}{#2}.{#3}}
\newcommand{\OutputV}[3]{\ActOutV{#1}{#2}.{#3}}
\newcommand{\TauP}[1]{\ActTau . {#1}}
\newcommand{\intextchoice}[2]{\Choice{\TauP{#1}}{\TauP{#2}}} 

\newcommand{\intaction}{\tau}
\newcommand{\red}{\st{\intaction}}

\newcommand{\ie}{{\em i.e.}\xspace}
\newcommand{\sts}[2]{\ensuremath{\langle #1, #2 \rangle}}

\DeclareUnicodeCharacter{2208}{$\in$}
\DeclareUnicodeCharacter{2203}{$\exists$}
\DeclareUnicodeCharacter{2200}{$\forall$}
\DeclareUnicodeCharacter{2113}{$\ell$}
\DeclareUnicodeCharacter{03B1}{$\alpha$}
\DeclareUnicodeCharacter{27F6}{$\longrightarrow$}
\DeclareUnicodeCharacter{2227}{$\wedge$}
\DeclareUnicodeCharacter{2228}{$\vee$}
\DeclareUnicodeCharacter{2261}{$\equiv$}
\DeclareUnicodeCharacter{3BC}{$\mu$}
\DeclareUnicodeCharacter{2260}{$\neq$}
\DeclareUnicodeCharacter{27F9}{$\Longrightarrow$} 
\DeclareUnicodeCharacter{03B7}{$\eta$} 
\DeclareUnicodeCharacter{03BB}{$\lambda$} 
\DeclareUnicodeCharacter{03C4}{$\tau$} 
\DeclareUnicodeCharacter{21D3}{$\Downarrow$}
\DeclareUnicodeCharacter{227C}{$\preccurlyeq$} 
\DeclareUnicodeCharacter{2081}{${}_1$} 
\DeclareUnicodeCharacter{2082}{${}_2$} 
\DeclareUnicodeCharacter{2091}{${}_{\textrm{e}}$} 
\DeclareUnicodeCharacter{219B}{$\nrightarrow$} 
\DeclareUnicodeCharacter{2286}{$\subseteq$} 
\DeclareUnicodeCharacter{2291}{$\sqsubseteq$} 
\DeclareUnicodeCharacter{2205}{$\emptyset$} 
\DeclareUnicodeCharacter{227E}{$\precsim$} 
\DeclareUnicodeCharacter{25B7}{$\triangleright$} 
\DeclareUnicodeCharacter{2194}{$\leftrightarrow$} 
\DeclareUnicodeCharacter{2250}{$\doteq$} 
\DeclareUnicodeCharacter{228E}{$\uplus$} 
\DeclareUnicodeCharacter{27FF}{$\rightsquigarrow$} 
\DeclareUnicodeCharacter{22CD}{$\simeq$} %
\DeclareUnicodeCharacter{2216}{$\setminus$} %
\DeclareUnicodeCharacter{2093}{$_x$} %
\DeclareUnicodeCharacter{2913}{$\downarrow_i$} %
\DeclareUnicodeCharacter{22D6}{$\lessdot$} %
\DeclareUnicodeCharacter{2AB7}{$\precapprox$} %
\DeclareUnicodeCharacter{2A7D}{$\leq$} %

\renewcommand{\mustset}[2]{#1 \mathrel{\opMustset} #2}

\newcommand{\ltsof}[1]{{\mbox{\small{\sf lts}}(#1)}}

\newcommand{\AxCmpl}{\axiom{AxCmpl}}

\newcommand\id{\mathit{id}}

\RequirePackage{tikz}
\RequirePackage{tikz-cd} 

\usetikzlibrary{calc,arrows,shapes,decorations.pathmorphing,decorations.markings,%
                backgrounds,positioning,arrows.meta,fit,decorations.pathreplacing}
\pgfdeclarelayer{bg}
\pgfsetlayers{bg,main}

\tikzstyle{distr} = [circle,fill=black!20,draw=black,thick,minimum size=2mm,
                     node distance=10mm and 12mm,inner sep=0pt]
\tikzstyle{dot} = [circle,fill=black!20,draw=black,thick,minimum size=1mm,
                     node distance=10mm and 12mm,inner sep=1pt]

\tikzstyle{state} = [rectangle,rounded corners,draw=black,thick,
                     minimum size=5mm, node distance=10mm and 12mm,
                     inner sep=3pt]
\tikzstyle{legent} = [node distance=10mm and 12mm, inner sep=2pt]
\tikzstyle{action} = [auto]
\tikzstyle{from} = [<->, shorten <=1pt, >=stealth',semithick]
\tikzstyle{timeto} = [->>, shorten >=1pt, >=stealth',semithick]
\tikzstyle{to} = [->, shorten >=1pt, >=stealth',semithick]
\tikzstyle{todistr} = [-, shorten >=1pt, >=stealth',semithick]
\tikzstyle{distrto} = [->, decorate, decoration={snake,pre length=1mm,post length=1mm}, shorten >=1pt, >=stealth',semithick]
\tikzstyle{tosqunder} = [to,rounded corners,swap,
                         to path={ (\tikztostart.south)
                                   -- ++(0,-.5)
                                   -- ($(\tikztotarget.south) + (0,-.5)$)
                                   \tikztonodes
                                   -| (\tikztotarget.south) }]
\tikzstyle{loop wnw} = [loop,looseness=7,out=185,in=175]
\tikzstyle{loop nee} = [bend right,looseness=5,out=135,in=100]
\tikzstyle{loop ene} = [bend left,looseness=5,out=-80,in=-45]
\tikzstyle{loop ese} = [loop,looseness=7,out=5,in=-5]
\tikzstyle{loop sse} = [loop,looseness=12,out=260,in=320]

\definecolor{pamblue}{rgb}{.78, .90, .98}

\def\bendamnt{13}

\tikzset{
  vertex/.style={text centered, fill, color=black, circle, inner sep=.7pt},
  state/.style={},
  every label/.style={label distance=-2pt},
  every label/.append style={font=\scriptsize},
  labelle/.style={%
    postaction={ decorate,
      decoration={ markings, mark=at position .5 with \node #1;}}}}

\newrobustcmd{\tilediag}[8]{
\begin{tikzpicture}
  \node[vertex, draw] (a) at (-1,0) [label=left:$#1$] {};
  \node[vertex, draw] (d) at (1,0) [label=right:$#2$] {};
  \node[vertex, draw] (b) at (0,1) [label=above:$#3$] {};
  \node[vertex, draw] (c) at (0,-1) [label=below:$#4$] {};

  \begin{pgfonlayer}{bg}
    \fill[pamblue] (a.center) to [bend left=\bendamnt] (b.center)
                   to [bend left=\bendamnt] (d.center)
                   to [bend left=\bendamnt] (c.center)
                   to [bend left=\bendamnt] (a.center);
  \end{pgfonlayer}
  \begin{scope}[line width=.45pt,decoration={
       markings,
       mark=at position 0.56 with {\arrow{Stealth}}}
     ]
     \draw[postaction={decorate}, labelle={[above left]{\footnotesize $#5$}}] (a.center) to [bend left=\bendamnt] (b.center);
     \draw[postaction={decorate}, labelle={[above right]{\footnotesize $#6$}}] (b.center) to [bend left=\bendamnt] (d.center);
     \draw[postaction={decorate}, labelle={[below left]{\footnotesize $#7$}}] (a.center) to [bend right=\bendamnt] (c.center);
     \draw[postaction={decorate}, labelle={[below right]{\footnotesize $#8$}}] (c.center) to [bend right=\bendamnt] (d.center);
   \end{scope}
\end{tikzpicture}
}

\tikzcdset{every label/.append style = {font = \scriptsize}} 

\newrobustcmd{\squarediag}[8]{
\begin{tikzpicture}

  \node[vertex, draw] (a) at (-1,0) [label=left:$#1$] {};
  \node[vertex, draw] (d) at (1,0) [label=right:$#2$] {};
  \node[vertex, draw] (b) at (0,1) [label=above:$#3$] {};
  \node[vertex, draw] (c) at (0,-1) [label=below:$#4$] {};

  \begin{scope}[line width=.45pt,decoration={
       markings,
       mark=at position 0.56 with {\arrow{Stealth}}}
     ]
     \draw[postaction={decorate}, labelle={[above left]{\footnotesize $#5$}}] (a.center) to [bend left=\bendamnt] (b.center);
     \draw[postaction={decorate}, labelle={[above right]{\footnotesize $#6$}}] (b.center) to [bend left=\bendamnt] (d.center);
     \draw[postaction={decorate}, labelle={[below left]{\footnotesize $#7$}}] (a.center) to [bend right=\bendamnt] (c.center);
     \draw[postaction={decorate}, labelle={[below right]{\footnotesize $#8$}}] (c.center) to [bend right=\bendamnt] (d.center);
   \end{scope}
\end{tikzpicture}
}

\newrobustcmd{\mytilediag}[8]{
\begin{tikzpicture}
  \node (eq) at (0,0) {\Large$=$};
  \node[vertex, draw] (b) at (-1,0) [label=left:$#1$] {};
  \node[vertex, draw] (c) at (1,0) [label=right:$#2$] {};
  \node[vertex, draw] (a) at (0,1) [label=above:$#3$] {};
  \node[vertex, draw] (d) at (0,-1) [label=below:$#4$] {};

  \begin{pgfonlayer}{bg}
    \fill[pamblue] (a.center) to [bend right=\bendamnt] (b.center)
                   to [bend right=\bendamnt] (d.center)
                   to [bend right=\bendamnt] (c.center)
                   to [bend right=\bendamnt] (a.center);
  \end{pgfonlayer}
  \begin{scope}[line width=.45pt,decoration={
       markings,
       mark=at position 0.99 with {\arrow{Stealth}}}
     ]
     \draw[postaction={decorate}, labelle={[above left]{\scriptsize $#5$}}] (a.center) to [bend right=\bendamnt] (b.center);
     \draw[postaction={decorate}, labelle={[above right]{\scriptsize $#7$}}] (a.center) to [bend left=\bendamnt] (c.center);
     \draw[postaction={decorate}, labelle={[below left]{\scriptsize $#6$}}] (b.center) to [bend right=\bendamnt] (d.center);
     \draw[postaction={decorate}, labelle={[below right]{\scriptsize $#8$}}] (c.center) to [bend left=\bendamnt] (d.center);
   \end{scope}
\end{tikzpicture}
}

\newrobustcmd{\mysquarediag}[8]{
\begin{tikzpicture}
  \node (eq) at (0,0) {\Large$\Longrightarrow$};
  \node[vertex, draw] (b) at (-1,0) [label=left:$#1$] {};
  \node[vertex, draw] (c) at (1,0) [label=right:$#2$] {};
  \node[vertex, draw] (a) at (0,1) [label=above:$#3$] {};
  \node[vertex, draw] (d) at (0,-1) [label=below:$#4$] {};
  
  \begin{scope}[line width=.45pt,decoration={
       markings,
       mark=at position 0.99 with {\arrow{Stealth}}
       }
     ]
     \draw[postaction={decorate}, labelle={[above left]{\scriptsize $#5$}}] (a.center) to [bend right=\bendamnt] (b.center);
     \draw[postaction={decorate}, labelle={[above right]{\scriptsize $#7$}}] (a.center) to [bend left=\bendamnt] (c.center);
     \draw[postaction={decorate}, labelle={[below left]{\scriptsize $#6$}}] (b.center) to [bend right=\bendamnt] (d.center);
     \draw[postaction={decorate}, labelle={[below right]{\scriptsize $#8$}}] (c.center) to [bend left=\bendamnt] (d.center);
   \end{scope}
\end{tikzpicture}
}

\title{A uniform characterisation of the (a)synchronous must-preorder}

\author{Giovanni Bernardi}%
      {Université Paris Cité, CNRS, IRIF, F-75013, Paris, France}%
      {\texorpdfstring{gio@irif.fr}{}}%
      {0009-0008-3653-3040}{}%
\author{Hugo Férée}%
      {Université Paris Cité, CNRS, IRIF, F-75013, Paris, France}%
      {\texorpdfstring{feree@irif.fr}{}}%
      {0000-0003-3103-5612}{}
\author{Gaëtan Lopez}%
      {Université Paris Cité, CNRS, IRIF, F-75013, Paris, France}%
      {\texorpdfstring{glopez@irif.fr}{}}%
      {0009-0002-3987-8999}{}
 
\authorrunning{G. Bernardi, H. Férée, G. Lopez}

\funding{This work has received support under the program ``Investissement d'Avenir'' launched by the French Government and implemented by ANR, with the reference ``ANR‐22‐CMAS-0001, QuanTEdu-France''. This work is supported by the European Union through the MSCA SE project QCOMICAL (Grant Agreement ID: 101182520).}

\ccsdesc[500]{Theory of computation~Operational semantics}
\ccsdesc[500]{Theory of computation~Program verification}
\ccsdesc[500]{Theory of computation~Distributed computing models}
\keywords{Contextual preorders, Full abstraction, Big step semantics, Constructive reasoning}

\begin{document}

\maketitle

\begin{abstract}
  In the setting of message passing software,
  De Nicola and Hennessy \mustpreorder 
  defines when a program improves on 
  another one.
  Since this preorder does not come equipped
  with a viable proof method, using it requires an
  alternative relation that characterises it.

  The literature presents at least {\em four different definitions} of
  such alternative preorders, depending on whether communication is
  synchronous or asynchronous and on whether there is value-passing
  or not. The existence of these different definitions complicates
  the overall theory, hinders the development of tools, and,
  upon the whole, suggests a lack of understanding of the properties
  necessary and sufficient to reason on the \mustpreorder.
  
  This paper presents the first alternative characterisation
  that works at least in all the four settings mentioned above.  We
  achieve this result thanks to an axiomatic approach
  that is calculus independent, by highlighting the role of blocking and
  non-blocking actions, and by introducing the
  novel notion of {\em label abstraction}.
  Label abstractions capture the essence of safe substitutivity
  \wrt interactions, and they let us obtain a unique proof of
  soundness and a unique proof of completeness that work in all the mentioned settings.
  We believe this generalises and simplifies the overall theory,
  while letting us present the existing results in a uniform way.

  Our proofs are constructive and our result is entirely
  mechanised in Rocq.
\end{abstract}
\section{Introduction}
\label{sec:introduction}

Since Morris seminal thesis \cite{morrisphd},
the study of contextual preorders is a standard
problem in the theory of programming languages. 
These preorders formalise the intuition that
a program~$Q$ can replace a program~$P$ if
no external observer can detect a difference
between them via (some form of) interaction.
Following \cite[pag. 51]{morrisphd},
one writes $P \ctxleq Q$ if for every context~$C[-]$,
if~$C[P]$ shows some observables 
then~$C[Q]$ exhibits analogous observables.

The definition of~$\ctxleq$ is simple, but
it does not entail a general proof method.
This is because,
given two programs~$P$ and $Q$, to show that $P \ctxleq Q$
one has to check an infinity of contexts.
As a consequence,
a direct proof of $P \ctxleq Q$ may be hard to find,
and, to make things worse, it may be
designed ad-hoc for~$P$ and~$Q$.
In other terms, reasoning on any other~$P'$ and~$Q'$
may require designing an entirely new proof.
This modus operandi does not scale, and thus, 
once a preorder~$\ctxleq$ is defined,
the challenge is to define an alternative relation, say~\altleq,
that is sound (i.e. ${\altleq} \subseteq {\ctxleq}$),
complete (i.e. ${\ctxleq} \subseteq {\altleq}$),
that provides a viable proof method,
and whenever possible also a decision procedure.
Soundness is essential for~\altleq\ to be of any use,
and completeness ensures that the preorder is nontrivial.
If~$\altleq$ is both sound and complete than it is
an {\em alternative characterisation} of~$\ctxleq$.

In the setting of synchronous message-passing the
most studied contextual preorders are probably De Nicola and
Hennessy~$\opMay$ and~$\opMust$ preorders
~\cite{DBLP:journals/tcs/NicolaH84,DBLP:books/daglib/0066919,DBLP:journals/iandc/HennessyI93,DBLP:journals/tcs/Hennessy02,DBLP:journals/iandc/RensinkV07,DBLP:journals/lmcs/DengGHM08,DBLP:conf/esop/KoutavasH11,DBLP:journals/fac/DengT12,DBLP:conf/birthday/Glabbeek22,DBLP:conf/esop/BernardiCLS25}.
In spite of their expressivity~\cite{DBLP:journals/tcs/Abramsky87},
compositionality (\cite[Proposition
  3.15]{Hennessy1993ModelPiCalculus}),
equational theories~\cite{DBLP:journals/iandc/HennessyI93},
co-inductive proof method~\cite{DBLP:journals/jacm/AcetoH92,DBLP:journals/mscs/BernardiH16},
decidability results \cite{DBLP:journals/iandc/KanellakisS90,DBLP:conf/aplas/BonchiCPS13,DBLP:journals/scp/BernardiF18}, and of
the insights they give on semantic subtyping 
\cite{DBLP:journals/corr/BernardiH13}, these preorders have been
studied essentially only using pen-and-paper. 
Recently, though,  the \mustpreorder has been characterised for
asynchronous \CCS in Rocq~\cite{DBLP:conf/esop/BernardiCLS25},
and in principle this result paves the way for the development of correct
by construction tools to employ the \mustpreorder in practice.
The results of \cite{DBLP:conf/esop/BernardiCLS25}, though, hold only
for an asynchronous communication model, and hence
neither can they be used off-the-shelf to reason on synchronous
calculi, nor do they let us port directly the proofs of all
the elegant properties mentioned above to the asynchronous semantics.
In fact, the results of \cite{DBLP:conf/esop/BernardiCLS25} do not even
account for an essential trait of programming languages: value-passing.
This contrasts with the empirical needs for tools to reason on
processes either with or without value-passing and in either
  asynchronous or synchronous communication models,
while minimising the amount of proofs to conceive,
understand, and explain (and the code to mechanise and maintain).

From an extensional standpoint, the different semantics impact the contextual
preorder under analysis. For instance, in synchronous process calculi all
actions are {\em blocking} while in asynchronous ones output actions
are {\em non-blocking}, and this asymmetry gives different 
(in)equational theories
~\cite{boudol:inria-00076939,DBLP:conf/ecoop/HondaT91,DBLP:journals/tcs/AmadioCS98,DBLP:books/daglib/0018113}.
To formulate a concrete example, following
\cite{DBLP:journals/iandc/CleavelandDSY99,DBLP:conf/esop/DengGMZ07,NunezLlana2008}
we let~$\envTEST$ denote a set of tests and we define the
\mustpreorder $\testleq$
by letting
$  \serverA \testleq \serverB %
\text{ if for every } \test \in \envTEST \wehavethat %
\Must{\serverA}{\test} \text{ then } \Must{\serverB}{\test},
$
where $\Must{\serverA}{\test}$ means that the program $\serverA$ satisfies the test $\test$
in every maximal execution.
We let \VCCS denote value-passing \CCS (augmented with $\Unit$ to
write successful tests), and we let \VACCS denote the {\em asynchronous}
variant of \VCCS~\cite{DBLP:conf/fsttcs/CastellaniH98,DBLP:journals/iandc/BorealeNP02,DBLP:conf/esop/BernardiCLS25}.
The syntactic definitions can be found in \Cref{appendix:VCCS}.
  
  \begin{example}[\coqVACCSEx{}]
  	
    \label{ex:lcc}
    \label{ex:linear-constant}
    
    We define two programs to show that~${\testleq[\VACCS]} \not\subseteq {\testleq[\VCCS]}$.
    Consider the ``copy-cat'' $\idVACCS = \lcc{a}$ that 
    writes back on the communication channel~$a$ the value that it reads on it.
    We have $\idVACCS \Ntestleq[\VCCS] \Nil$
    because in \VCCS we can write the test $a!w.\Unit$
    (for some value~$w$), which
    is satisfied by $\idVACCS$ and not by~$\Nil$.
    On the other hand, letting~$\testeq[\VACCS]$ denotes the equivalence relation induced by~$\testleq[\VACCS]$,
    we have $\idVACCS \testeq[\VACCS] \Nil$,
    because in \VACCS no test exists that distinguishes~$\idVACCS$ and~$\Nil$. 
    As second example let $\constVACCSv{v} = a?(x).a!v.\Nil$.
    Again the limited power of asynchronous tests ensures
    that     $\constVACCSv{v} \testleq[\VACCS] \Nil$
    for all value~$v$. Instead in \VCCS we have~$\constVACCSv{v} \Ntestleq[\VCCS] \Nil$,
    because~$\constVACCS$ satisfies~$a!w. \Unit$, while~$\Nil$ does not.
    \hfill\qed 
\end{example}
The difference between preorders shown in \Cref{ex:lcc} has justified a
development of ad-hoc theories for asynchronous calculi to treat in a
suitable manner non-blocking outputs.
This is the case for asynchronous bisimilarity
\cite{DBLP:journals/tcs/AmadioCS98,DBLP:books/daglib/0004377,DBLP:books/daglib/0018113},
and also for the \mustpreorder: the alternative characterisations defined in
\cite{DBLP:conf/esop/BernardiCLS25,DBLP:journals/iandc/BorealeNP02} 
are not sound w.r.t. testing preorders in synchronous calculi.
A similar remark holds also for theories about calculi with
and without value-passing. For instance, the authors of
\cite{DBLP:journals/iandc/HennessyI93} show that the alternative
preorder defined in~\cite{DBLP:journals/tcs/NicolaH84}
is not complete in presence of value-passing.


From a theoretical standpoint, the fact that there is an ad-hoc
characterisation for each semantics (synchronous or asynchronous, with
or without value passing) suggests that the literature does not capture yet
the essential features of the \mustpreorder.

From a practical standpoint, as testing preorders formalise
notions of correct substitutivity, it only makes sense for
the correctness of their characterisations to be trustworthy.
Given the size of actual software (millions of lines of code),
the only way to arrive at trustworthy and practically viable
characterisations is via machine checked proofs.
In turn, this means that ad-hoc alternative characterisations
imply the development and maintenance of distinct but similar
proofs and verified tools, which is both costly and bad
engineering practice. The question is thus:
{\em is there a unique characterisation of the \mustpreorder for
both synchronous and asynchronous settings with and without value-passing?}

{\bfseries Main contribution.}
In this paper we answer the above question positively:
\Cref{thm:main-result} proves that the \mustpreorder coincides with
the standard extensional preorder for partial functions
{\em in all four settings} (i.e. synchronous /asynchronous, with/without value-passing)
via
\begin{inparaenum}[(1)]
  \item a {\em unique} big-step interpretation, denoted  $\interp{}{ \server } $,
    which treats programs as partial functions with codomain $\parts{\parts{X}}$ for some set~$X$; and
  \item a suitable preorder on $\parts{\parts{X}}$.
\end{inparaenum}
We establish this result in the Rocq theorem-prover via {\em a unique} proof
that works in all the four settings. 
Moreover, the proof of \Cref{thm:main-result} shows that the argument
for synchronous calculi is {\em the same} argument as for
the asynchronous ones, with merely the additional hypothesis
that all actions are blocking.

Our arguments are constructive. This ensures that a proof of $\serverA \testleq \serverB$ is essentially a function that transforms any proof of $\Must{\serverA}{\test}$  (\ie that ~$\serverA$ guarantees the liveness of the test~$\test$),
into a proof of $\Must{\serverB}{\test}$. \Cref{thm:main-result}
may thus shed light on recent program logics like Trillium
\cite{DBLP:journals/pacmpl/TimanyGSHGNB24}.

{\bfseries Technical challenges.}
To arrive at \Cref{thm:main-result} we have to overcome three technical issues.
First, we have to find the relation between the elements of the set $X$
and the behaviour of programs. 
Indeed, one of the differences among the existing characterisations is
the concrete definition of the (elements of the) set~$X$.
To reason on \CCS the set~$X$ contains all visible actions,
while for \ACCS it contains only the outputs, and for
\VCCS $X$ is the set of 
visible actions stripped of their values (\ie $a?$ and $a!$).
The characterisation of the \mustpreorder
given in
~\cite{DBLP:journals/tcs/NicolaP00}, instead, relies on (a subset of)
visible actions {\em keeping} 
their values.
The other two challenges have to do with the (proof of) completeness
of our characterisation. The second technical challenge is that
we need to find sufficient properties of \LTSs to write the tests
necessary for the proof to work.
Since our result also accounts for value-passing, the axioms
of \cite{DBLP:conf/esop/BernardiCLS25} do not suffice, and
we need to find a more general set of properties.
The third challenge is that the proof relies on some
elements of~$\parts{\parts{X}}$
to construct tests that distinguish two programs whenever their
interpretations are not in the extensional preorder.
A priori, though, the sets in $\parts{\parts{X}}$ could contain
an infinite amount of actions (for instance $a?v$ for any value $v$).
Given that we are working with constructive logic, and that we are
aiming for effective proof methods based on the alternative preorder,
we thus want to only work with (finite sets of) finite sets.

\subsection{Further contributions and paper outline}
In \Cref{sec:preliminaries} we introduce~$\ltsmultiset{\NBLabels}$,
\ie the class of \LTSs whose elements have visible actions in some set~\Act,
and whose parameter~\NBLabels is a possibly empty subset of~\Act.
We require all the transitions labelled by actions in~\NBLabels
to respect the axioms for asynchrony given in \Cref{fig:axioms}.
They define what it means for an action to be {\em non-blocking},
and they are a mild generalisation of Selinger axioms for asynchrony
\cite{DBLP:conf/esop/BernardiCLS25,DBLP:conf/concur/Selinger97} and of the
axioms for forwarding \cite{DBLP:conf/esop/BernardiCLS25,DBLP:conf/ecoop/HondaT91,DBLP:books/daglib/0018113}.
Essentially this means that the \LTSs in $\ltsmultiset{\NBLabels}$ represent
the behaviours of first-order programs {\em and} of the shared medium that
they use to communicate. The medium modelled by the axioms is a
multiset, and if $\NBLabels = \emptyset$ then the medium does not
exist, all actions are blocking, and communication is synchronous.
We remark that the notion of \textit{value} is extrinsic
to that of LTS, and thus even stating characterisations à la
\cite{DBLP:journals/iandc/HennessyI93} in our axiomatic treatment
is not possible. The ``right'' alternative characterisation must thus
be independent from the nature of values.
In \Cref{sec:alt-preorder} we define the interpretation of every
state, say~$\server$, of \LTSs belonging to $\ltsmultiset{\NBLabels}$ as
a partial function $\interp{}{ \server } : \Actfin \rightharpoonup
\parts{\parts{X}}$ for some set~$X$ depending on $\labs$, where
the symbol $\labs$ is an argument that denotes a {\em label
  abstraction}. To the best of our knowledge, this is a novel notion
which we present in \Cref{def:label-abstraction}.
Label abstractions capture the intrinsic dependence between labels and transitions
in \LTSs, while guaranteeing enough finiteness conditions to
solve the technical difficulties sketched above.
For instance, in the early style LTS of \VACCS,
for all $v,w \in \Val$, a process~$P$ performs an
input~$a?v$ if and only if it performs also the input~$a?w$.
This ensures that to perform one interaction with $P$
it suffices to offer it an output $a!z$ {\em regardless}
of the actual value $z$. In other terms,
\begin{inparaenum}[(1)]
\item values can be abstracted in outputs, and
\item outputs on the same channel should have the same abstraction.
\end{inparaenum}
To reason on \VACCS, we will use a label abstraction that ensures
exactly these relations between labels and transitions.
Moreover label abstractions guarantee that the actual codomain
of the interpretation is the set of finite sets $\parts{\fparts{X}}$,
which can be effectively inspected to build a test.
As the program behaviours that we work with are finite-image,
the definitions will actually ensure that our proofs manipulate the set~$\fparts{\fparts{X}}$.
Given their role in our proof, we see 
the definition label abstraction and of their properties
as the second contribution of this work.

In \Cref{sec:soundness} we prove soundness, \ie that for every LTS $\lts$ with states $\States$,
set of tests~$\envTEST$, label abstraction~$\labs$ that abstracts $(\lts, \envTEST)$,
and every $\serverA, \serverB \in \States$, we have that
$\interp{}{ \serverA } \altleq \interp{}{ \serverB }$ implies $\serverA \testleq \serverB$.
In \Cref{sec:completeness} we overcome the third technical issue and prove completeness.
More specifically,  we introduce the set of axioms $\AxCmpl$ on {\em both}~$\labs$
and~$\envTEST$, which suffices to prove that
if $\labs$ is finitary, 
then $\serverA \testleq \serverB$ implies $\interp{}{ \serverA } \altleq \interp{}{ \serverB }$.
In view of this, the third contribution of this paper is the set
of axioms~$\AxCmpl$ (\Cref{tab:properties-functions-to-generate-tests}).

Before outlining our proofs, in \Cref{sec:applications} we apply our unique alternative
characterisation to the (early style) \LTSs of the calculi \CCS, \ACCS, \VCCS, and \VACCS, thereby proving that \Cref{thm:main-result}
not only implies but also provides a common generalisation of the
behavioural characterisations
in~\cite{DBLP:books/daglib/0066919,DBLP:journals/iandc/HennessyI93,DBLP:conf/esop/BernardiCLS25},
along with a new characterisation of~$\testleq[\VACCS]$.
To the best of our knowledge, this is the first constructive and machine checked
such result for~$\testleq[\VACCS]$, and we see it as the fourth
contribution of this paper.
One technique we use to arrive at the results about concrete calculi
may also be worth attention. First, we parameterise the notion of forwarder
\cite{DBLP:conf/esop/BernardiCLS25,DBLP:conf/ecoop/HondaT91,DBLP:books/daglib/0018113},
and in particular the \boom axiom, on the set $\NBLabels$.
Second, in \Cref{lem:liftFW-correct-in-general} we show that the
transformation~$\liftFWSym$ that adds forwarding to the LTS of processes,
not only preserves the \mustpreorder, but also collapses to a graph isomorphism
if the calculus is synchronous. This suggests {\em why} results that hold for
synchronous calculi can be adapted to asynchronous ones via forwarding.

We summaries and conclude this paper in \Cref{sec:conclusion}, where we also
discuss related and future works.
\section{Preliminaries}
\label{sec:preliminaries}

Following the literature, we fix a countable set $\Act$ of
so-called ``visible'' or ``external'' actions, ranging over it with
$\aa, \ab, \ac$. We let $\tau$ denote an \emph{invisible} action
not in $\Act$ that represents internal computation
and denote the set of all actions
by $\Labels \eqdef \Effects \cup {\set{ \tau }}$,
ranged-over by~$\anyaction$ and~$\anyactionb$.
We also assume the existence of an involutive {\em duality} function
over $\Act$ and denote by $\co{\mu}$ the dual of $\mu$.
Additionally, as outlined in the Introduction,
we let~$\NBLabels$
denote a possibly empty subset of~$\Effects$
whose elements are \emph{non-blocking} actions,
and we range over it with $\nba, \nbb, \nbc$.
We let $\BLabels \eqdef \Effects \setminus \NBLabels$. This set contains
all the {\em blocking} actions, and we range over it with $\bla, \blb, \blc$.
Intuitively, blocking actions can be used to ``prefix'' a program behaviour
(\ie they can be used as guards in \CCS prefix).
We give concrete example of actions at the end of this section.



%
%
A {\em labelled transition system} (LTS) is a triple
$\lts = \LTS{ \States }{ \st{} }{ \Labels }$ where~$\States$ is the set
of states, and ${\st{}} \subseteq \States \times \Labels \times \States$
is the transition relation.
Given an LTS $\lts$, if its set of labels is a singleton 
we say that $\lts$ is a {\em state transition system} and we omit to write its third component,
i.e. the set of labels.
We write  $\state \st{\alpha} \stateA$ to mean that
$(\state,\alpha,\stateA)~\in~{\st{}}$ and $ \state \st{ \alpha }$
to mean $\exists \stateA \suchthat \state \st{\alpha} \stateA$.
We write $\lts_X$ to state that $X$ is the set of states in $\lts$.

        When convenient, we write our statements using only one LTS $\lts$, 
        because given two \LTSs~$\lts_X$ and~$\lts_Y$ one can always obtain
        a single LTS via the set-theoretic union of their components.

\renewcommand{\States}{X}
When reasoning on programs operationally, it is 
common to work up-to strong bisimilarity.
To allow such up-to reasoning in our mechanisation,
we assume that for every LTS~$\lts_{\States}$, the set~$\States$
is endowed with some equivalence relation~$\simeq$
that is compatible with the transition relation,
i.e. that satisfies the following property:
$
\forall
\serverA, \serverB, \serverC \in \States \wehavethat
\serverA \simeq \serverB \st{\alpha}~\serverC$
implies that
$
\exists \serverD \in \States \suchthat  \serverD  \st{\alpha} \serverA' \simeq \serverC.
$
\noindent%
Since the identity over~$\States$ enjoys this
property 
we assume the existence of~$\simeq$ without any loss of generality.

\renewcommand{\StatesA}{P}
\renewcommand{\StatesB}{T}

We proceed to define the \mustpreorder. Given two LTSs $\lts_\StatesA$ and $\lts_\StatesB$, the STS of their composition
is $\csys{\lts_\StatesA}{\lts_\StatesB} \eqdef \sts{\StatesA \times \StatesB}{\st{}_{PT}}$,
where $\st{}_{PT}$ is obtained via the 
following rules~(\coqToSet{parallel_gLts}):\\
\stauserver: if $ \serverA \red \serverB$ then $\csys{\serverA}{\test} \red \csys{\serverB}{\test}$;\\
\stauclient: if $ \test \red \test' $ then $\csys{\server}{\test} \red \csys{\server}{\test'}$;\\
\scom: if $ \serverA \st{ \co{\aa} } \serverB$ and $\testA \st{ \aa } \testB$ then $ \csys{\serverA}{\test} \red \csys{\serverB}{\test'}$.\\
These rules restrict the standard definition (see for instance \cite{DBLP:conf/fossacs/AcetoI99})
to permit only internal moves in the resulting STS. This suffices for our purposes.
From now on we denote the states of $\csys{\lts_\StatesA}{\lts_\StatesB}$
with $\csys{\server}{\test}$ instead of $(\server, \test)$.

To formally define the \mustpreorder,
given an STS $\sts{ \csys{\server}{\test} }{ \st{ }}$,
a \emph{computation} of $\csys{\server}{\test}$ is a 
sequence of transitions that starts from $\csys{\server}{\test}$.
A computation is
{\em maximal} if either it cannot be extended or it is infinite.
We also assume a decidable predicate~$\goodSym$ over tests,
and we assume it is preserved by non-blocking actions and by
the equivalence~$\simeq$:
if $\client \st{ \eta } \client'$ then $\good{ \client}$ iff $\good{ \client'}$,
and if $\client \simeq \client'$ then $\good{ \client }$ iff $\good{ \client'}$.
The previous works on testing for asynchronous settings rely, possibly implicitly, on these assumptions,
where the equivalence is usually the structural one \cite{DBLP:conf/fsttcs/CastellaniH98,DBLP:journals/iandc/BorealeNP02,DBLP:journals/jlp/Hennessy05,DBLP:conf/esop/BernardiCLS25}.


\begin{definition}[\textrm\mustpreorder, \coqMustE{}]
  \label{def:must-extensional}  
  \label{def:testleq}
  We write $\Must{\server_0}{\client_0} $ if in every maximal        
  computation $\csys{ \server_0 }{
  \client_0 } \st{ } \csys{ \server_1 }{ \client_1 } \st{ } \csys{
  \server_2 }{ \client_2 } \st{ } \ldots $
there exists a state $\csys{ \server_i}{ \client_i}$ such that
$\good{\client_i}$.

Given an LTS of tests $\envTEST$,
  we write $ \server \testleq \serverB$
  whenever for every $\client \in \envTEST$
  we have that
  $\Must{\server}{\client}$
  implies
  $\Must{\serverB}{\client}$.
\end{definition}
\noindent
Note that $\testleq$ is a function monotonically decreasing wrt
its variable~$\envTEST$.
For instance, using the syntax of \VCCS, ${\testleq[\emptyset]} = {\testleq[\set{\Nil}]} = {\testleq[\set{\Unit}]} = {\testleq[\set{\rec{a!\unit.x}}]}$
is the total relation, while we have
$a!\unit.\Nil \Ntestleq[\set{ a?(x).\Unit }] b!\unit.\Nil$
because the test $a?(x).\Unit$ has some distinguishing power, in the sense that it
can be passed by at least one program and that it is failed by at least one program.
We have for example $\Must{a!\unit.\Nil}{ a?(x).\Unit}$ and $\NMust{b!\unit.\Nil}{ a?(x).\Unit}$.

\begin{figure*}[t]
\hrulefill
\\
\begin{center}

\begin{tabular}{c@{\hskip10pt}c}

  \hypertarget{nb-delay}{}
  \begin{tabular}{l@{\hskip4pt}c@{\hskip0pt}l}
    \begin{tikzcd}
      \stateA \arrow[r,"\nonblockingaction"]& \stateB \arrow[d,"\anyaction"]\\
      & \stateC
    \end{tikzcd}
    &$\Rightarrow$&\,\,
    \begin{tikzcd}
      \stateA\arrow[r,"\nonblockingaction"]\arrow[d,"\anyaction"]&\stateB\arrow[d,"\anyaction"]\\
      \stateD\arrow[r,"\nonblockingaction"]&\stateC
    \end{tikzcd}
  \end{tabular}
&

\hypertarget{nb-confluence}{}
\begin{tabular}{l@{\hskip4pt}c@{\hskip0pt}l}
\begin{tikzcd}
\stateA\arrow[r,"\nonblockingaction"]\arrow[d,"\extaction"]& \stateB\\
\stateC&
\end{tikzcd}
&
$\Rightarrow$
&
\begin{tikzcd}
\stateA\arrow[r,"\nonblockingaction"]\arrow[d,"\extaction"] & \stateB \arrow[d,"\extaction"]\\
\stateC\arrow[r,"\nonblockingaction"]&\stateD
\end{tikzcd}
\end{tabular}

\\
\nbdelay
&
\nbconfluence

\\[7pt]

\hypertarget{nb-determinacy}{}
\begin{tabular}{l@{\hskip4pt}c@{\hskip0pt}l}
\begin{tikzcd}
  \stateA\arrow[r,"\nonblockingaction"]\arrow[d,"\nonblockingaction"]&\stateB
  \\
  \stateC&
\end{tikzcd}
&
$\Rightarrow$
&
\, $\stateB=\stateC$
\end{tabular}
&
\hypertarget{nb-inv-determinacy}{}
\begin{tabular}{l@{\hskip4pt}c@{\hskip0pt}l}
	\begin{tikzcd}
		&\stateB\arrow[d,"\nonblockingaction"]
		\\
		\stateC\arrow[r,"\nonblockingaction"]&\stateA
	\end{tikzcd}
	&
	$\Rightarrow$
	&
	\, $\stateB=\stateC$
\end{tabular}
\\[5pt]
\nbdeterminacy
&
\nbdeterminacyinv
\\[5pt]
\hypertarget{fwd-feedback}{}
\begin{tabular}{lc@{\hskip0pt}l}
  \begin{tikzcd}
    \server_1 \arrow[r, "\nonblockingaction"] & \server_2 \arrow[d, 		"\co{\nonblockingaction}"] \\
    & \server_3
  \end{tikzcd}
  &
  $\Rightarrow$
  &
  \,\, $\server_1 \st{ \tau } \server_3$ or $\server_1 = \server_3$
\end{tabular}
&
\hypertarget{boomerang}{}
\begin{tikzpicture}
  \node (\serverA) {$\exists \stateA . \state$};
  \node[right=+40pt of \serverA] (\serverA') {$\stateA$};
  
  \path[->]
  (\serverA) edge [bend left] node[above] {$\co{\nonblockingaction}$}  (\serverA')
  (\serverA') edge [bend left] node[above] {$\nonblockingaction$}  (\serverA);
\end{tikzpicture}
\\[5pt]
\fwdfeedback & \boom
\\[7pt]
\end{tabular}

\hypertarget{nb-tau}{}
\begin{tabular}{l@{\hskip4pt}c@{\hskip0pt}lcl}
  \begin{tikzcd}
    \stateA\arrow[r,"\nonblockingaction"]\arrow[d,"\ActAsTau"]&\stateB\\
    \stateC&
  \end{tikzcd}
  &$\Rightarrow$&\,\,
  \begin{tikzcd}
    \stateA\arrow[r,"\nonblockingaction"]\arrow[d,"\ActAsTau"]&\stateB\arrow[d,"\ActAsTau"]\\
    \stateC\arrow[r,"\nonblockingaction"]&\stateD
  \end{tikzcd}
  & or\ &
  \begin{tikzcd}
    \stateA\arrow[r,"\nonblockingaction"]\arrow[d,"\ActAsTau"]&\stateB\arrow[dl,"\co{\nba}"]\\
    \stateC&
  \end{tikzcd}
\end{tabular}
\\[5pt]
\nbtau
\end{center}
\caption{Basic axioms for non-blocking actions. 
  We omitted the following quantifications from each axiom:
  $\forall \anyaction \in \Labels, \forall \nonblockingaction \in \NBLabels, \forall \extaction \in \Effects \text{ and } \extaction \neq \nba$ (\coqgLTSOBA{gLtsOba}, \coqgLTSFW{gLtsFW}).
}
\label{fig:axioms}
\hrulefill
\end{figure*}

We conclude this section introducing the
classes of \LTSs over which we prove our main result,
consider the axioms in \Cref{fig:axioms}.
We define \ltsmultiset{\NBLabels} as the set of all \LTSs satisfying
those axioms. We comment swiftly on two of them.
The axiom \nbdelay states that the actions in~$\NBLabels$
can always be delayed, and hence they are non-blocking.

The axiom \boom formalises that each state in the \LTSs
models the behaviour of a program {\em and of a mutable state}.
If we think of blocking actions as inputs and non-blocking
actions as outputs, then the axiom means that in any state~$p$
if~$p$ can perform an output, then the datum communicated
must have been previously input.
Further explanations of these axioms can be found in
\cite{DBLP:conf/esop/BernardiCLS25}, here we remark that
the LTS of multisets satisfies this axiom.

\begin{example}[\coqMutiSetLTS{lts_multiset_step}]
	
\label{ex:MO}
Let~$\MO$ denote the set of all {\em finite} multisets of non blocking actions, 
for instance $\varnothing, \mset{ \nba }, \mset{ \nba,  \nba
}, \mset{ \nba,  \nbb,  \nba,  \nbb} \in \MO$.
The symbols $M, N, \ldots$ range over~$\MO$, and we denote
with~$\uplus$ the multiset union.
The way in which a multiset is manipulated by the environment is
formalised by the following
transition rules for every $\nba \in \NBLabels$:
$ M \st{ \co{\nba} }_{\sf m} M \uplus \mset{ \nba }$ and $
M \uplus \mset{ \nba } \st{ \nba }_{\sf m} M$.
The first rule models that a multiset can perform any input (i.e. it
is input-enabled),
thereby changing its state by adding the value read.
The right rule states the dual: a multiset can output some datum that
it contains,
thereby removing it. It is plain that $\LTS{\MO}{\st{}_{\sf
m}}{\NBLabels}$ enjoys \boom.
Note that this LTS does not contain $\tau$ transitions. This should be obvious:
a communication medium performs no computations, i.e. it is merely passive and
manipulated by programs. 
The LTS that we just defined is depicted in \cite[Fig.~3.2]{DBLP:series/txtcs/Fokkink07} for $\Val = \set{0,1}$.%
\hfill\qed
\end{example}

We will use also the following class of \LTSs whose transitions
describe the behaviour only of programs.
Let \agents{\NBLabels} contain all the \LTSs that satisfy the
axioms in \Cref{fig:axioms} except \fwdfeedback and \boom,
and that additionally satisfy the following axiom:\\
\nbfeedback: \hypertarget{nb-feedback}{$\forall \serverA,\serverB,\serverC,\nba \in \NBLabels. \ 
  \serverA \st{\eta} \serverB \st{\co{\eta}} \serverC$
  implies 
  $\serverA \st{ \tau } \serverC$}. 

This property means that if in a state of an LTS the program
can put a message in the communication (modelled via the action
$\eta$) and remove it (modelled by the dual $\co{\eta}$),
then the state can evolve via a $\tau$ move that models
the execution of the two side-effects that we just described.

\begin{example}
	\label{ex:concrete-LTSs}
        \label{ex:duality-in-input-output-model}
        We recall two \LTSs that belong to the class $\agents{\NBLabels}$ for suitable
        sets $\Act$ and $\NBLabels$.
        Let \Names and \Val be countable sets respectively of names and values.
	In the LTS of \CCS 
	the visible actions are $\Act = \Names \cup \Out$, where $\Out = \setof{ \co{a} }{ a \in \Names }$.
        In this case $\NBLabels = \emptyset$, \ie the calculus is synchronous, and the
        duality is the obvious bijection between $\Names$ and $\Out$.
	In the early style LTS of \VACCS the visible actions are
	$\ActV = \setof{
		\ActInV{a}{\somevalue} }{ a \in \Names, \somevalue \in \Val } \cup \OutV$,
                where
                $\OutV = \setof{ \ActOutV{a}{\somevalue} }{ a \in \Names, \somevalue \in \Val }$.
                In this case outputs are non-blocking, and thus we let $\NBLabels = \OutV$.
        The duality is defined letting 
	$\co{\ActOutV{\somechannel}{\somevalue}} \eqdef \ActInV{\somechannel}{\somevalue}$
	and
	$\co{\ActInV{\somechannel}{\somevalue}} \eqdef \ActOutV{\somechannel}{\somevalue}$
	for any $\somechannel \in \Names$, $\somevalue \in \Val$.
	\hfill\qed
\end{example}
\section{The alternative preorder}
\label{sec:alt-preorder}
  
We now introduce the necessary definitions to characterise the \mustpreorder.

Recall a standard notion from process algebra: given a LTS $(\lts, \server)$ its 
{\em ready set} is defined as $\readyset{ \server } \eqdef \setof{ \extaction
  \in \Effects }{\server \st{\extaction}}$, i.e. the set of all the
external actions that $\server$ may perform.
Now recall that $\BLabels = \Act \setminus \NBLabels$.
The first ingredient of our characterisation is a variant of ready-sets,
the {\em co-ready set}, which we define as
$  \coR{ \server } \eqdef \setof{ \bla \in \BLabels }{\server \st{\co{\bla}}} $
where $\st{}$ is the transition relation of~$\lts$.
The set $\coR{ \server }$ contains all the {\em blocking} actions that let
an external observer interact with~$\server$.
This definition is somewhat subtle. For instance,
in general $ \coR{
  \server } \neq \setof{ \co{\extaction} \in \Effects }{\server \st{
    \extaction }}$. To see why, recall the program~$\idVACCS$ of \Cref{ex:lcc}.
We have seen in \Cref{ex:concrete-LTSs} that in the early style LTS of $\VACCS$
the dual of inputs are outputs, and they are all non-blocking.
We therefore have $\coR{ \idVACCS } = \emptyset$.
On the other hand, the set $\setof{ \co{\extaction} \in \Effects }{\idVACCS \st{ \extaction }}$
is $\setof{ \ActOutV{ a }{v} }{v \in \Val, a \in \Names}$.
We now have enough notation to define label abstractions.


\begin{definition}[Label abstraction, \coqAbs{AbsAction}]
  \label{def:label-abstraction}
        A {\em label abstraction} $\labs$ for $\Act$
        is a pair of functions $(\latestSym: \BLabels \rightarrow \FinA, \laprogSym: \FinA \rightarrow \PreA)$,
        for arbitrary sets $\FinA$ and $\PreA$.

  We say that $\labs$ {\em abstracts} a pair $(\lts,\envTEST)$ of \LTSs if for every $\bla, \blb \in \BLabels$:
    \begin{enumerate}
    	\item \label{pt:test-action-abstraction-property}%
    	for every $\test \in \domOfTest$
    	if
    	$\latest{\bla} = \latest{\blb}$ then
    	$\bla \in \readyset{\test} \implies \blb \in \readyset{\test}$;
    	\item \label{pt:program-co-action-abstraction-property}%
    	for every $\server \in \domOfServer$ 
    	if 
    	$\deltadualphifinite{\bla} = \deltadualphifinite{\blb}$
        then 
    	$\latest{\bla} \in \latest{\coR{\serverA}} \implies
    	\latest{\blb} \in \latest{\coR{\serverA}}$.
    \end{enumerate}
    
    Finally, when $\labs$ abstracts a pair $(\lts,\envTEST)$,
    we say that $\labs$ is {\em finitary} if 
    for every $\server$  in $\lts$ the set
    $\map{\coRsym}{\labs}{\server} \eqdef \deltadualphifinite{\coR{ \server }}$ is finite.
\end{definition}
\noindent
\renewcommand{\PreA}{X_{\labs}}
From now on we write~$\PreA$ to denote the codomain of the function $\laprogSym$ in a label abstraction~$\labs$.


To characterise the \mustpreorder we reason on how a test, \ie a context,
can interact with a program that is in a {\em waiting state}~\cite{DBLP:conf/ecoop/HondaT91},
\ie that cannot perform any $\tau$-transition. Some authors call these
states ``potential deadlocks''
\cite{DBLP:journals/jacm/Hennessy85,DBLP:conf/esop/BernardiCLS25}.

\renewcommand{\States}{Y}
Given an LTS $\lts_\States$, the waiting states of a~$\server \in \States$
after executing (non-deterministically) a trace~$\trace$ are given by the
following function~(\coqFiniteImage{wt_refuses_set}):
\begin{equation}
\label{def:waiting-states}
  \begin{array}{lll}
    \wsSym[\lts] & : & \text{\States} \rightarrow (\Actfin
    \rightharpoonup \parts{ \States })
    \\
    \wsSym[\lts](\serverA)
    &
    \eqdef
    &
    \lambda \trace . 
    \setof{\serverB}{\serverA \wt{ \trace } \serverB \stable }
  \end{array}
\end{equation}
Note that infinite sequences of $\tau$'s in the first argument given to
the function~$\wsSym$ make it diverge, this means that the function is {\em partial}.
For instance if $\hat{\lts} = \LTS{ \set{\Omega} }{ (\Omega,\tau,\Omega) }{ \set{\tau} }$,
then $\ws{\hat{\lts}}{\Omega}{\varepsilon}$ is not defined.
In order to reason on traces that belong to the domain of
the function $\lambda \trace . \ws{\lts}{\serverA}{ \trace }$ we need
to ensure that it is defined, \ie that it the rooted LTS does not
lead the computation of~$\wsSym$ into a divergence.
To this end, we say that {\em $\serverA$ converges along the trace~$\trace$},
denoted~$\serverA \cnvalong \trace$ whenever, by performing~$\trace$,
$\serverA$ cannot reach a state that can perform an infinite sequence of $\tau$
transitions.
The set $\ws{\lts}{\serverA}{\trace}$ is well-defined constructively only if $\server\cnvalong\trace$.
When $\lts$ is clear from the context, we write $\ws{}{\serverA}{\trace}$ {\em in lieu} of $\ws{\lts}{\serverA}{\trace}$.
Note also that since we work with finite-image \LTSs, if  $\ws{}{\serverA}{\trace}$ is defined,
then it is finite.

\renewcommand{\States}{Y}
Given an~$\labs$ and an~LTS $\lts$ with states in~$Y$
we define our interpretation thus:
\begin{equation}
  \label{def:co-acceptance-set}
  \begin{array}{lll}
    \interpSym{} & : & \text{\States} \rightarrow (\Actfin \rightharpoonup \parts{\parts{ X_\labs }})
    \\
  \interp{}{\serverA}
  & \eqdef & \lambda \trace . \map{\coRsym}{\labs}{\ws{}{\serverA}{ \trace }}
  \end{array}
\end{equation}
\noindent
The {\em total} function $\interpSym{}$ interprets
rooted \LTSs into a partial functions from finite
sequences of visible actions to sets of sets.
For instance, in the LTS $\hat{\lts}$ the interpretation $\interp{}{\Omega}(\varepsilon)$ is not defined.
This is a departure from the existing literature that studied
acceptance (sets) with pen-and-paper. For instance in
\cite{DBLP:books/daglib/0066919,DBLP:journals/corr/BernardiH15,DBLP:journals/mscs/BernardiH16}
we have $\acc{\rec{\VRP}}{\varepsilon} = \emptyset$,
\ie the so-called acceptance set of $\rec{\VRP}$ is defined.
In the present setting, instead, the interpretation of $\rec{\VRP}$ is {\em tout-court} undefined.
The definition of $\interpSym{}$ suggests that we should see it as a form of
big-step semantics: given a program $\server$ and a trace $\trace$, if the value
$\interp{\lts}{\server}(\trace)$ is defined, it tells which abstractions of
blocking actions may be used to interact with the derivatives of~$\server$ after
it has executed the trace~$\trace$.


Defining an alternative preorder for~$\testleq$ is now a matter
of comparing partial functions, which we do via~$\extleq$, \ie
the standard extensional point-wise order on them.
We first need to define a preorder on the codomain of the
functions under discussion, i.e. $\parts{\parts{X_\labs}}$.
For any set $Z$ and $A,B \in \parts{\parts{Z}}$ we let $A \ll B$
whenever $\forall x \in B $ there exists an element $y \in A$ such
that $y \subseteq x$, and given two partial functions $f$ and $g$ with
codomain $\parts{\parts{Z}}$ we let $f \extleq g$ whenever
for all $x \in \dom{ f }$ we have
$x \in \dom{ g }$ and 
$f(x) \ll g(x)$. 
Since $\trace \in \dom{ \interp{}{\serverA} } $ if and only if $\serverA \cnvalong \trace$,
when applied to our interpretations, the preorder~$\extleq$
coincides with the following one.
\begin{definition}[\coqAbs{bhv_pre}]
	\label{def:accset-leq}%
        For every LTS $\lts$,
        every programs $\serverA,\serverB$ in $\lts$
        and label abstraction $\labs$,
        we write $\serverA \altleq[\labs] \serverB$ whenever
        for every $\trace \in \Actfin$
        if $\serverA \cnvalong \trace$
        then
        \begin{inparaenum}[(i)]
	\item $\serverB \cnvalong \trace$, and                 
	\item 
	  $
	  \interp{}{\serverA}( \trace )
          \ll
          \interp{}{\serverB}( \trace )$.
	\end{inparaenum}%
\end{definition}
\noindent
The definition of the interpretation $\interpSym[-]{}$ and of the preorder~$\ll$
generalise Hennessy's definition of (preorder over) acceptance sets
\cite{DBLP:conf/programm/Hennessy82,DBLP:books/daglib/0066919,DBLP:journals/iandc/HennessyI93},
which in turn coincides with the preorder on the Smyth power domain \cite{smyth_power_1978}.
Our main result if the following fact.

\begin{theorem}[\coqEquivAS{equivalence_fw_bhv_acc_ctx}]
	\label{thm:main-result}
	For every $\lts \in \ltsmultiset{\NBLabels}$,
	every two programs $\serverA, \serverB$ in~$\lts$,
        every $\envTEST \in \agents{\NBLabels}$,
        and every~$\labs$ that abstracts $(\lts,\envTEST)$ we have%
	\begin{enumerate}[(1)]
		\item\label{pt:main-result-soundness}
		if
		$\serverA \altleq[\labs] \serverB$
		then
		$\serverA \testleq \serverB$; 
                and
		\item\label{pt:main-result-completeness}
		if $(\envTEST, \labs) \models \AxCmpl$ and~$\labs$ is
                finitary then 
		$\serverA \testleq \serverB$
		implies
		$\serverA \altleq[\labs] \serverB$.
	\end{enumerate}
\end{theorem}
\noindent
This theorem lays bare that to characterise~$\testleq$ it is necessary and sufficient
\begin{inparaenum}[(i)]
\item to abstract only blocking actions, and
\item to use the relations between~$\Act$ and the
structure of the transitions in the underlying \LTSs
identified in \Cref{def:label-abstraction}.
\end{inparaenum}
The hypothesis that $\labs$ is finitary is necessary to prove \Cref{thm:main-result}(\ref{pt:main-result-completeness}) in constructive logic.
In the next examples we show that the hypothesis that~$\labs$ abstracts $(\lts_B,\envTEST)$
is necessary for \Cref{thm:main-result} to hold.
\begin{example}
  \label{ex:labs-abstracts-necessary-soundness}
  \label{we-need-reasonable-abstraction-for-value-passing}
  Let $\labs$ be a label abstraction where~$\latestSym$ is a constant function.
  If we use tests in the LTS of \CCS, $\labs$ does not satisfy \Cref{def:label-abstraction}(\ref{pt:test-action-abstraction-property}), as it maps every label to the same image.
  Intuitively, this means that $\labs$ abstracts {\em too much}.
  For instance, consider $\test \eqdef \Inputon{b}.\Success$.
  We have that $\test \st{ \Inputon{b} }$ and
  $\test \Nst{ \Inputon{a} }$, while instead $\latest{\Inputon{a}} = \latest{\Inputon{b}}$.
  Owing to this excessive loss of information (on the behaviour of tests) of $\labs$,
  we prove that ${\altleq[\labs]} \not \subseteq  {\testleq[\set{\test}]}$.
  Consider the two \CCS programs $p \eqdef \Choice{\Outputon{a}.\Nil}{\Outputon{b}.\Nil}$
  and $q \eqdef \Outputon{a}.\Nil$.
  Napkin calculation shows that $p \altleq[\labs] q $, on the other hand
  we have $p  \Ntestleq[\set{\test}] q $ because 
  $\Must{ p }{ \test }$ and $\NMust{ q }{\test}$.
  \hfill\qed
\end{example}

Moreover, in general the abstraction needs to lose some information for
\Cref{thm:main-result}(\ref{pt:main-result-completeness}) to hold.
We show this in \Cref{ex:labs-abstracts-necessary-completness-VACCS},
which we give at the end of the next section.

We outline the proof of \Cref{thm:main-result} in \Cref{sec:soundness}
and \Cref{sec:completeness}, and now we
conclude the section proving the crucial lemma on label abstraction.
Consider the following axiom, which expresses a form of enabledness,
\ie that a state can perform the dual of every non-blocking actions, 
\enabled:
$
\forall \serverA \wehavethat
\forall \nba \in \NBLabels \wehavethat
\exists \serverB \wehavethat \serverA \st{ \co{\nba} } \serverB
$.
This is only ``half'' of the \boom axiom, and indeed if $\lts \models
\boom$ then $\lts \models \enabled$.

\begin{lemma}[\coqSoundAS{communication_enabled}]
  \label{lem:labs-and-scom}
  For every LTS $\lts_{\domOfServerA}$, %
  every LTS $\envTEST \in \agents{\NBLabels}$,
  every LTS~$\lts_{\domOfServerB} \models \enabled$,
  for every program $\serverA \in \domOfServerA$, $\serverB \in \domOfServerB$, for every test $\test$ in $\envTEST$ such that $\map{\coRsym}{\labs}{\serverA} \subseteq \map{\coRsym}{\labs}{\serverB}$, if $\csys{\serverA}{\test} \red$ is derived using $\scom$,
  then we have $\csys{\serverB}{\test} \red$ using $\scom$.
\end{lemma}
  \begin{proof}
    Fix two $\serverA$, $\serverB$ such that $\map{\coRsym}{\labs}{\serverA} \subseteq \map{\coRsym}{\labs}{\serverB}$ ($\star$).
    Take a test $\test \in \domOfTest$ such that $\csys{\serverA}{\test} \red$ and suppose that this silent transition was derived using $\scom$,
	hence it exists some external action $\extaction \in \Effects$ such that $\serverA \st{\co{\extaction}}$ and $\test \st{\extaction}$.
	Now we distinguish whether the action $\extaction$ performed by the test is a blocking action or not.
        
	If $\extaction \in \NBLabels$, i.e. the action is non blocking, 
	hence the hypothesis $\lts_{\domOfServerB} \models \enabled$ implies $\serverB \st{ \co{\extactionA} }$. We now apply rule \scom\ to derive $\csys{ \serverB }{ \test} \red$.

	Instead, if $\extaction \in \BLabels$ then the definition of $\coR{-}$
	implies that
	$\extaction \in \coR{\serverA}$, thus 
	by ($\star$)
	$\laprog{\latest{\extaction}} \in \map{\coRsym}{\labs}{\serverB}$.
	By the definition of $\mapSym{\coRsym}{\labs}$, there exists a $\extaction' \in \Effects$ such that $\deltadualphifinite{\extaction} = \deltadualphifinite{\extaction'}$ and $\latest{\extaction'} \in \coRlatest{\serverB}$.
	Since $\laprog{\latest{\extaction}} = \laprog{\latest{\extaction'}}$, Definition~\ref{def:label-abstraction}(\ref{pt:program-co-action-abstraction-property})
	imply that $\latest{\extaction} \in \coRlatest{\serverB}$,
	and again the definition of $\coRlatestSym$
	ensures that there exists a $\extaction'' \in \Effects$ such that $\serverB \st{\co{\extaction''}}$ and $\latest{\extaction} = \latest{\extaction''}$.
	Since $\latest{\extaction} = \latest{\extaction''}$, Definition~\ref{def:label-abstraction}(\ref{pt:test-action-abstraction-property}) and $\test \st{\extaction}$ imply $\test \st{\extaction''}$.
	We now apply rule \scom\ to obtain the required $ \csys{ \serverB }{ \test } \red$.
\end{proof}
\noindent
In essence, this lemma shows that $\map{\coRsym}{\labs}{\serverA}$
captures exactly how a test may not only interact with $\serverA$, but more precisely how it can
distinguish interactions with $\serverA$
from the ones with other programs.
This is not obvious, as $\map{\coRsym}{\labs}{-}$ in defined using
only blocking actions, and moreover, the label abstraction could lose
essential information. 
The hypothesis that $\labs$ abstracts $(\lts, \envTEST)$ is indeed
necessary to prove \Cref{thm:main-result}. This is indirectly shown in
\Cref{ex:labs-abstracts-necessary-soundness}, and will be shown in
\Cref{ex:labs-abstracts-necessary-completness-VACCS}.

In the next section we show how to apply \Cref{thm:main-result}
to obtain proof methods for concrete calculi.

\begin{table}[t]
  \hrulefill
  \begin{center}
    \scalebox{.9}{
	\begin{tabular}{%
            |@{\hskip 2pt}c@{\hskip 2pt}%
            |@{\hskip -2pt}c@{\hskip -2pt}%
            |@{\hskip -2pt}c@{\hskip -2pt}%
            |@{\hskip 2pt}c@{\hskip 2pt}%
            |@{\hskip 2pt}c@{\hskip 2pt}%
            |}
		\hline
		& \textbf{$\latestSym$} & \textbf{$\laprogSym$} & $\NBLabels$ & \begin{tabular}{c}previous\\ characterisation\end{tabular}
                  \\
		\hline
		\CCS & $\id$ & $\id$ & $\emptyset$ & \cite[Theorem 4.4.6]{DBLP:books/daglib/0066919} \\
		\hline
		\ACCS & $\id$ & $\id$ & $\Out$ & \cite[Theorem 1]{DBLP:conf/esop/BernardiCLS25} \\
		\hline
		\VCCS & %
                \begin{tabular}{r@{\hskip 1pt}c@{\hskip 1pt}l}%
                  $\ActOutV{c}{v}$ & $\mapsto$ &  $\ActOutV{c}{v}$ \\
                  $\ActInV{c}{v}$ &  $\mapsto$ & $\Inputon{c}$
                \end{tabular}
                &
                \begin{tabular}{r@{\hskip 1pt}c@{\hskip 1pt}l}%
                  $\ActOutV{c}{v}$ & $\mapsto$ & $\Outputon{c}$ \\
                  $\Inputon{c}$ & $\mapsto$ & $\Inputon{c}$
                \end{tabular}
                &  $
                \emptyset$ & \cite[Equation (1)]{DBLP:journals/iandc/HennessyI93}
                \\
		\hline
		\VACCS & $\ActInV{c}{v} \mapsto \Inputon{c}$ & $\id$ & $\OutV$ & --\\
		\hline
	\end{tabular}
        }
        \end{center}
	\caption{Concrete functions and values of $\NBLabels$ to define abstractions for four calculi.}
        \label{tab:label-abstractions}
	\label{tab:label-abs-for-V(A)CCS}
        \hrulefill
\end{table}

\section{Applications} 
\label{sec:applications}
We outline how to apply \Cref{thm:main-result} to obtain the
alternative characterisation of the \mustpreorder for \VCCS, \VACCS,
and their variants without value passing.
We need first to introduce an auxiliary function.
The transitions of each state in the \LTSs of~$\ltsmultiset{\NBLabels}$ describe the behaviour of
programs together with the communication medium, while
the transitions of the \LTSs in $\agents{\NBLabels}$
represents only programs.
Not surprisingly, programs (behaviours) can be made
to interact with a mutable communication medium (a multiset),
thereby mapping the elements of $\agents{\NBLabels}$ into
$\ltsmultiset{\NBLabels}$. Crucially for our aims,
this mapping preserves the \mustpreorder.

Recall from \Cref{ex:MO} the definition of the LTS of finite multisets.
\begin{definition}[\coqToFW{FW_gLts}]
	\label{def:liftFW}

	We let
        $\liftFWSym : \agents{\NBLabels} \longrightarrow
        \ltsmultiset{\NBLabels}$ that maps any $\lts_X$
        to $\LTS{X \times \MO}{ \stfw{} }{ \Labels }$,
	where the transition relation
        is defined by the following rules:\\
        \stLeft: %
        if $\server \st{\anyaction} \server'$
        then $\server \triangleright \someMSET \stfw{\anyaction} \server' \triangleright \someMSET$;\\
        \stRight: %
        if $\someMSET \st{\anyaction}_{\sf m} \someMSET'$
        then $\server \triangleright \someMSET \stfw{\anyaction} \server \triangleright \someMSET'$;\\
	\stInter: %
        if $\server \st{\co{\nba}} \server' \quad \someMSET \st{\nba}_{\sf m} \someMSET'$
        then $\server \triangleright \someMSET \stfw{\ActTau} \server' \triangleright \someMSET'$.
\end{definition}
\noindent
The implementation of the function~$\liftFWSym$ does indeed what its
type states, \ie $\liftFW{\lts}$ satisfies all the axioms in
\Cref{fig:axioms}, moreover it gives us an insight on {\em
  synchronous} theories. First, a lemma.

\begin{lemma}[\coqToFW{FW_gLtsObaFW}]
  \label{lem:liftFW-correct-in-general}
  For every LTS~$\lts \in \agents{\NBLabels}$ 
\begin{inparaenum}[(1)]
\item $\liftFW{\lts} \in \ltsmultiset{\NBLabels}$,
\item\label{pt:liftFW-isomorphism}%
  if $\NBLabels = \emptyset$ then $\liftFW{\lts}$ is isomorphic to $\lts$.
\end{inparaenum}
\end{lemma}
\begin{proof}
  The proof of the first point can be
  found in \cite{DBLP:conf/esop/BernardiCLS25}.
  To prove the second point, note that since $\NBLabels = \emptyset$,
  to construct the relation $\stfw{}$ the only rule of
  \Cref{def:liftFW}
  that can applied is \stLeft,
  and thus~$\liftFW{\lts}$ contains exactly the same transitions of~$\lts$.
The isomorphism between the graph~$\liftFW{\lts}$ and the graph $\lts$ is the
function $\mathit{first}(p,M)~=~p$.
\end{proof}
\noindent
Consider now \Cref{lem:liftFW-correct-in-general}(\ref{pt:liftFW-isomorphism}).
In synchronous settings, \ie where all actions are blocking, the
function $\liftFWSym$ is essentially the identity!
This means that all results on synchronous calculi
are de facto given up-to an application of $\liftFWSym$.
On the other hand, in asynchronous settings, the function
adds enough transitions to model a medium that enjoys the axiom \enabled,
(in particular, for asynchronous variants of \CCS the resulting LTS is
input-enabled).

Note that in asynchronous \CCS and its variants, the memory
is modelled via {\em syntax} and the LTS of the calculus,
and thus it is {\em not} input-enabled.
In the literature this limitation is systematically fixed by applying to programs
a transformation equivalent to $\liftFWSym$. This is the case for
bisimulation \cite{DBLP:conf/ecoop/HondaT91,DBLP:journals/tcs/AmadioCS98},
testing \cite{DBLP:conf/fsttcs/CastellaniH98} and is used also
on Petri Nets \cite{DBLP:conf/birthday/BaldanBGV15}.

We follow these examples and obtain the characterisation of the
\mustpreorder for a series of calculi by proving that $\liftFWSym$
preserves the \mustpreorder, and applying \Cref{thm:main-result}
to the \LTSs (of the calculi) amended via~$\liftFWSym$.
To this end, we need a technical axiom:\\
\axiom{Finite-NB-Chains}: $\forall \server \wehavethat \exists n \in \N \wehavethat
\forall (\nba_1 , \nba_2 ,\ldots, \nba_n) \in \NBLabels^{n} \wehavethat
\neg ( \server \st{\nba_1} \cdots \st{\nba_2})$.\\
In any LTS satisfying this axiom, there is no infinite sequence of non-blocking actions. This
  amounts to the reasonable assumption that the content of the medium is finite.
\begin{lemma}[\coqLift{lift_fw_ctx_pre}]
	
  \label{lem:liftFW-preserves-testleq}
  Let $\lts \in \agents{\NBLabels}$ and
  let~$\envTEST \in \agents{\NBLabels}$.
  If $\lts \models \axiom{Finite-NB-Chains}$ then
  for every $\serverA, \serverB$ in~$\lts$,
      $\serverA \testleq \serverB$ if and only if $\liftFW{\serverA} \testleq \liftFW{\serverB}$.
\end{lemma}

In the next proposition we denote with $\ACCS,\CCS,\VACCS,\VCCS$
the early style LTS of the corresponding calculus and we let $\lts$ range
over these \LTSs. We obtain the proof method for each possible $\lts$
via the label abstraction~$\labs_{\lts}$ as defined in \Cref{tab:label-abs-for-V(A)CCS}.
\renewcommand{\ltsof}[1]{#1}
\begin{proposition}[\coqACCSCorolarry{},\coqCCSCorolarry{},\coqVACCSCorolarry{},\coqVCCSCorolarry{}]
  \label{cor:characterisation-for-all-CCS}
  For any $\lts \in \set{\ACCS,\CCS,\VACCS,\VCCS}$, 
  for every $\serverA, \serverB \in \lts$,
  $\serverA \testleq[\lts] \serverB$ if and only if
  $\liftFW{\serverA} \altleq[\labs_\lts] \liftFW{\serverB}$.
\end{proposition}
\begin{proof}
  For every $\lts$ we have that
  $ \ltsof{\lts} \models \axiom{Finite-NB-Chains}$.
  Moreover we have that
  ${\ltsof{\ACCS}}~\in~\agents{\Out}$,
  ${\ltsof{\VACCS}} \in \agents{\OutV}$,
  and 
  ${\ltsof{\lts}} \in \agents{\emptyset}$ for $\lts \in \set{\CCS,\VCCS}$.
  \Cref{lem:liftFW-correct-in-general} thus implies that every 
  $\liftFW{\ltsof{\lts}}$ is in the class $\ltsmultiset{\NBLabels}$ for a suitable $\NBLabels$.
  The proofs of these facts boil down to the one given in
  the full version of \cite{DBLP:conf/esop/BernardiCLS25}.

  We have $ (\ltsof{\lts},\labs_\lts) \models \AxCmpl$
  and that $ \labs_\lts $ abstracts $(\ltsof{\lts},\ltsof{\lts})$ and therefore we conclude
  applying \Cref{thm:main-result}.
\end{proof}
\noindent
To the best of our knowledge,
\Cref{cor:characterisation-for-all-CCS}
is the first constructive characterisation
of the \mustpreorder for $\VACCS$.
Thanks to this proposition we can prove code
transformations correct, for instance in \VACCS we have the following inequality:
for every boolean expression~$\be$, and $M \in \MO $, and $  \forall \serverA,\serverB \in \VACCS$, we have that
\begin{equation}
  \label{eq:general-code-hoisting}
  \tau.(\serverA \Par \Pi M) \extc \tau.(\serverB \Par \Pi M)
  \testleq[\VACCS]
  (\Pi M) \Par (\tau.\serverA + \tau.\serverB)
\end{equation}
This is a restricted form of distributivity of parallel over sum,
and it models code hoisting, which is a typical optimisation
made by compilers.

We conclude remarking that
if a label abstraction does not lose information (\ie abstract) at all,
then the completeness result may not hold.
\begin{example}
	\label{ex:labs-abstracts-necessary-completness-VACCS}
        Let $\labs \eqdef (\id,\id)$ and let~$\lts$ denote the
        early style LTS of \VACCS, where $\Val = \set{0 , 1}$.
        Following \Cref{def:label-abstraction}, $\labs$ abstracts
        $(\lts,\lts)$.
        Note, though, that since $\labs$ acts as the identity, it loses no
        information at all.
        Recall from \Cref{ex:linear-constant} that $\constVACCSv{0} \testleq[\VACCS] \idVACCS$.

	Now consider the trace~$\trace = \ActInV{\somechannel}{1}$. Both $\constVACCSv{0}$ and
        $\idVACCS$ converge along $\trace$, and calculations show that
        $\interp{}{\liftFW{\idVACCS}}(\trace)=\set{\set{\ActInV{\somechannel}{1}}}$.
        Thanks to forwarding we also have that 
        $\interp{}{\liftFW{\constVACCSv{0}}}(\trace) =
        \set{\set{\ActInV{\somechannel}{0}}}$,
        hence
        $\interp{}{\liftFW{\constVACCSv{0}}}(\trace) \not\ll
        \interp{}{\liftFW{\idVACCS}}(\trace)$ and
        $\liftFW{\constVACCSv{0}} \Naltleq[\labs] \liftFW{\idVACCS}$.
        Completeness does not hold: ${\testleq[\VACCS]} \not\subseteq {\altleq[\labs]}$.\hfill\qed
\end{example}

\section{Soundness}
\label{sec:soundness}

\renewcommand{\States}{S}

As explained in \cite[Section 3.2]{DBLP:conf/esop/BernardiCLS25}, in a constructive
setting, \Cref{thm:main-result}(\ref{pt:main-result-soundness})
is a consequence of a more general property that pertains to LTS whose states are sets.
The first step is to adapt the interpretation given in \Cref{def:co-acceptance-set}
to these LTSs.  Given an LTS with states in~$\parts{ \States }$ for some~$\States$
consider the following interpretation:
\begin{equation}
  \label{def:interpretation-set}
  \begin{array}{lll}
    \interpSym{\sfset} & : & \parts{ \States } \rightarrow (\Actfin \rightharpoonup \parts{\parts{ X_\labs }})
    \\
  \interp{\sfset}{ Y }
  & \eqdef & \lambda \trace . \bigcup_{ \serverA \in Y } \interp{}{\serverA}( \trace )
  \end{array}
\end{equation}

Note that the interpretation given in \Cref{thm:main-result} when used on the LTS of sets
loses information on non-determinism, while the function  $\interpSym{\sfset}$ does not.
Consider the following example.

\begin{example}
  Fix two different $a,b\in\Names$ and let $p$ be the \CCS term $\intextchoice{\Outputon{a}.\Nil}{\Outputon{b}.\Nil}$. 
  Let $\lts = \ltsof{\CCS}$ and let~$\labs$ abstract $(\lts,\lts)$, which ensures that for $x = \deltadualphifinite{\Inputon{a}}$ and $y = \deltadualphifinite{\Inputon{b}}$ we have that
  $x \neq y$.
  We have $\interp{ }{ p }(\emptytrace) \neq \interp{\sfset}{ \set{ p } }(\emptytrace)$.
  To see why, observe that $\interp{ }{ p }(\emptytrace) = \set{\set{ x, y  }}$.
  On the other hand, if $\set{ p } \wt{} \someserverset \stable$, then $\someserverset = \set{\Outputon{a}.\Nil , \Outputon{b}.\Nil}$, and thus
  $\interp{\sfset}{ \set{p } }(\emptytrace) = \set{\set{ x }, \set{ y }}$.
  That is, $\interp{ }{ p }(\emptytrace)$ contains only one set with two different elements,
  thereby losing information on the non-determinism of $p$,
  while $\interp{\sfset}{ \set{p } }(\emptytrace)$ contains two different singletons.
  \hfill\qed
\end{example}

\renewcommand{\States}{P}
The interpretation in \Cref{def:interpretation-set} produces partial functions that require
convergence to be used. We do not need to change the convergence predicate, moreover we know
the following fact. In the following we let $\lts$ be an LTS with set of states~$\States$.
\begin{lemma}[\coqSetLTS{convergence_set_iff_convergence_forall}]
	\label{lem:conv-set-iff-forall-conv}
        Let $\someserverset \in \pparts{\States}$ and  $\trace \in \Actfin$,
        $\someserverset \cnvalong \trace$ if and only if $\forall \server \in \someserverset \suchthat \server \cnvalong \trace$.
\end{lemma}
We adapt the alternative preorder in the obvious manner.
\begin{definition}[\coqSoundAS{bhv_pre__x}]
  For every $X, Y \in \pparts{\States}$
  and $\labs$, 
  we write $X \altleqset[\labs] Y$
  whenever
  for every $\trace \in \Actfin$
  if $X \cnvalong \trace$
  then
	$Y \cnvalong \trace$ and            
	  $
	  \interp{\sfset}{X}( \trace )
          \ll
          \interp{\sfset}{Y}( \trace )$.
\end{definition}

To obtain the soundness result, we also generalise to sets both the $\opMust$ predicate and the contextual preorder.
Given a set of programs~$X$ and an LTS of tests~$\test$, we write $\mustset{ X }{\test}$ to mean that
$\forall \server \in X, \Must{\server}{\test}$ (\coqSoundAS{must_set_iff_must_for_all}),
and we write $X \testleqset Y$ whenever
for all $\test \in \envTEST$ if $\mustset{\someserversetA}{\test}$ then $\mustset{\someserversetB}{\test}$.

Given an LTS  $\lts = \LTS{\States}{\st{}}{ \Labels }$
we let $\toSet{\lts} = \LTS{\pparts{ \States }}{ \stset{} }{ \Labels }$,
where  the transition relation is defined for every $ X \in \pparts{ \States } $,
as $X \stset{ \alpha } \setof{ \serverB }{ \exists \serverA \in X \suchthat \serverA \st{\anyaction} \serverB }.$
This is the standard function to transform non-deterministic LTS into deterministic ones, already used for instance in \cite{DBLP:conf/avmfss/CleavelandH89,DBLP:journals/corr/abs-1302-1046}.
\begin{lemma}[\coqSoundAS{alt_set_singleton_iff}, \coqSoundAS{must_set_singleton_iff}]
  \label{lem:toSet-preserves-preorders-on-singletons}
  For every $\serverA, \serverB$ in $\lts$, 
  \begin{enumerate}[(1)]
  \item \label{alt-as-program-to-set}
    For every $\labs$, $\serverA \altleq[\labs] \serverB$
    if and only if $\toSet{\serverA} \altleqset \toSet{\serverB}$;		
  \item  \label{must-set-to-program}
    for every $\envTEST$, 
    $\serverA \testleq \serverB$ if and only if $\toSet{\serverA} \testleqset \toSet{\serverB}$.
  \end{enumerate}
\end{lemma}

\Cref{lem:toSet-preserves-preorders-on-singletons} and the following proposition are 
sufficient to obtain \Cref{thm:main-result}(\ref{pt:main-result-soundness}).

\begin{proposition}[Soundness, \coqSoundAS{soundnessx}]
	\label{prop:soundness}
        Let~$\lts \models \enabled$,
        let $\envTEST \in \agents{\NBLabels}$,
        and let $\labs$ that abstracts $(\lts,\envTEST)$.
        For every~$X, Y \in \pparts{P}$,
	if
	$X \altleqset Y$
	then
	$X \testleqset Y$.
\end{proposition}


\renewcommand{\StatesB}{Q}
The proof of \Cref{prop:soundness} is analogous to the
one of \cite{DBLP:conf/esop/BernardiCLS25}, except for
one key difference: the use of label abstractions.
The crucial lemma applied in the proof is the following one.
\begin{lemma}[\coqSoundAS{stability_nbhvleqtwo}]
	\label{lem:stability-nbhvleqtwo}
	Suppose that $\lts_\StatesB \models \enabled$,
        that $\labs$ abstracts $(\lts_\StatesB,\envTEST)$ for some $\envTEST$,
        and let $\lts_\StatesA$ be an LTS.
        For every 
        $\someserversetA \in \pparts{ \StatesA }$,
        $\someserversetB \in \pparts{ \StatesB } $
        such that
	$\someserversetA \altleqset \someserversetB$,
	and for every $\test \in \envTEST$
	if $\ungood{\test}$ and
        $\mustset{\someserversetA}{\test}$
	then
        $\forall \serverB \in \someserversetB \wehavethat \csys{\serverB}{\test} \red$.
\end{lemma}
\begin{proof}
  Fix $\serverB \in \someserversetB$.
        Suppose that~$\serverB$ and~$\test$ are stable (otherwise we can already conclude).
	Since $\serverB$ is stable we know that
	$\map{\coRsym}{\deltadualphifiniteSym}{\serverB} \in \interp{\lts_B}{ Y }( \emptytrace )$.
	The hypotheses $\ungood{\test}$ and $\mustset{\someserversetA}{\test}$ imply that
        $\someserversetA \cnvalong \varepsilon$, and thus $\someserversetA \altleqset \someserversetB$
        gives us $\serverA \in \someserversetA$ and $\serverA' \in \envServerA$ such that $\serverA \wt{ } \serverA' \stable$ and ($\star$) $\map{\coRsym}{\deltadualphifiniteSym}{\serverA'} \subseteq \map{\coRsym}{\deltadualphifiniteSym}{\serverB}$.	
	By definition, 
	$\someserversetA \wt{ } \someserversetA'$ for some
        set $\someserversetA'$ such that $ \serverA' \in \someserversetA' $,
        and so the hypothesis $\mustset{\someserversetA}{\test}$
        ensures that $\mustset{ \set{\serverA'} }{\test}$.
	As $\ungood{\test}$, $\mustset{\set{\serverA'}}{\test}$
        implies that $\csys{\set{ \serverA' }}{\test} \red$.
	Since $\test$ and $\set{ \serverA' }$ are stable
        (resp. by assumption and by definition)
        this $\tau$-transition must have been
	derived using \rname{s-com}, and so there
        exists a $\extactionA \in \Effects$ such that 
	$\set{ \serverA'}  \st{\co{\extactionA}}$ and $\test \st{\extactionA}$.
        The hypothesis that 
        $\lts_\StatesB \models \enabled$ and 
        that $\labs$ abstracts $(\lts_\StatesB,\envTEST)$ let us conclude via
        apply \Cref{lem:labs-and-scom}. %
\end{proof}
\begin{table}[t]
  \hrulefill\\
  There exists two functions
  $\testconvSym : \Actfin \rightarrow \envTEST$,
  and $\testaccSym  : \Actfin \times \fparts{\PreA} \rightarrow \envTEST$ 
  such that for all
  $\trace \in \Actfin$, for all
  $\extactionA, \extactionB \in \Act$,
  for all $\preactionsubset \subseteq \PreA$ and any $f \in\set{\testconvSym,\testaccSym}$:
\\[10pt]
	\begin{minipage}{300pt}%

          
		\begin{enumerate}[(1)]
			\item
			\label{test-ungood}
			$ \lnot \good{f(\trace)}$
			\item
			\label{test-next-step}
			$f(\extactionA . \trace) \st{\extactionA} f (\trace)$
			\item
			\label{test-tau-transition}
			$\extactionA \in \BLabels$ implies $f (\extactionA.\trace) \red $
			\item
			\label{test-reset-tau-path}
			$ \forall \ctest \in \domOfTest$, $\extactionA \in \BLabels$ and $f (\extactionA.\trace)
			\red \ctest$ implies $\good{ \ctest }$
			\item
			\label{test-follows-trace-determinacy}
			$\forall \ctest \in \domOfTest$, $\extactionA  \in \BLabels$ and $f(\extactionA . \trace)  \st{\extactionB} \ctest$ and $\extactionB = \extactionA$ implies $\ctest = f(s)$
			\item
			\label{test-side-effect-by-construction}
			$\forall \ctest \in \domOfTest$, $\extactionA \in \BLabels$ and $f(\extactionA . \trace)  \st{\extactionB} \ctest$ and $\extactionB \neq \extactionA$ implies $\good{ \ctest }$
		\end{enumerate}
	\end{minipage}
	\\[5pt]
	\begin{minipage}{300pt}
		\begin{enumerate}[({a}1)]
			\item
			\label{ta-does-no-tau}
			$\testacc{\emptytrace}{\preactionsubset} \Nst{\tau}$
			
			\item
			\label{ta-does-no-non-blocking-actions}
			$\forall \nonblockingaction \in \NBLabels, \testacc{\varepsilon}{\preactionsubset} \Nst{\nonblockingaction}$
			
			\item
			\label{ta-actions-are-in-its-gamma-set}
			$\forall \blockingaction \in \BLabels, \testacc{\varepsilon}{\preactionsubset} \st{\blockingaction} \test$ implies
			$\deltadualphifinite{\blockingaction} \in \preactionsubset$
			
			\item
			\label{ta-has-a-representative-transition-for-its-gamma-set}
			$\forall \preactionsubsetelem \in \preactionsubset$, $\exists \blockingaction \in \BLabels$ such that $\deltadualphifinite{\blockingaction} = \preactionsubsetelem$ and $ \testacc{\varepsilon}{\preactionsubset} \st{\blockingaction}$
			
			\item
			\label{ta-transition-to-good}
			$\forall \blockingaction \in \BLabels, \ctest \in \States,
			\testacc{\varepsilon}{\preactionsubset} \st{\blockingaction} \ctest$ implies $\good{\ctest}$
		\end{enumerate}
	\end{minipage}
	\\[5pt]
	\begin{minipage}{300pt}
		\begin{enumerate}[(c1)]
			\item \label{tc-does-no-external-action}
			$\forall \extaction \in \Act, \testconv{\varepsilon} \Nst{\extaction}$ \qquad
			\item \label{tc-can-compute}
			$\exists \ctest \in \domOfTest, \testconv{\varepsilon} \red \ctest$ \qquad
			\item \label{tc-computes-to-good}
			$\forall \ctest \in \domOfTest, \testconv{\varepsilon} \red \ctest$ implies
			$\good{ \ctest }$
			\\
		\end{enumerate}
	\end{minipage}
	
	\caption{Axioms for completeness (\coqCpltAS{test_spec}).}
	\hrulefill
	\label{tab:properties-functions-to-generate-tests}
\end{table}


\section{Completeness}
\label{sec:completeness}
We now focus on \Cref{thm:main-result}(\ref{pt:main-result-completeness}).
We cannot work in the powerset monad as we did in \Cref{sec:soundness},
because the function $\toSetSym$ does not preserve the axioms in \Cref{fig:axioms}.
\begin{example}
  \label{ex:toSet-breaks-axioms}
  We prove that $\toSet{\lts} \not\models \nbdelay$, where $\lts = \LTS{\ACCS}{\Acttau}{\st{}}$.
  Let $\serverA = \ActInOn{a}.\ActOutOn{a}.\Nil$, let $\serverB = \Prl{\ActOutOn{a}.\Nil}{\ActInOn{a}.\Success}$,
  and consider the transitions of the set $X = \set{ \tau.\serverA  + \tau.\serverB }$. We have that $ X \st{ \tau } \badset $,
  where $ \badset = \set{ \serverA, \serverB }$.
  By definition there exists the transition
  $\serverB \st{\ActOutOn{a}} \ActInOn{a}.\Success$,
  which implies that
  $\badset \st{\ActOutOn{a}} \set{\ActInOn{a}.\Success} \st{\ActInOn{a}} \set{\Success}.$
  To finish the argument, we show that 
   for every set $Z \wehavethat \badset \st{\ActInOn{a}} \cdot \st{\ActOutOn{a}} Z \text{ implies } Z \neq  \set{\Success}$. 
  To see this, observe that the transitions
  $\serverA \st{\ActInOn{a}} \ActOutOn{a}.\Nil$ and $\serverB \st{\ActInOn{a}} \Prl{\ActOutOn{a}.\Nil}{\Success}$,
  imply that
  $\badset \st{\ActInOn{a}} \set{\ActOutOn{a}.\Nil, \Prl{\ActOutOn{a}.\Nil}{\Success}}$.
  Since the transition relation in $\toSet{\lts}$ is deterministic,
  $ \badset \st{\ActInOn{a}} Y$ implies $Y = \set{ \ActOutOn{a}.\Nil, \Prl{\ActOutOn{a}.\Nil}{\Success} }$.
  This means that for all $Y$ such that $\badset \st{\ActInOn{a}} Y$ it must be the case that
  $Y \st{ \ActOutOn{a} } \set{ \Nil, \Success }$, and clearly $  \set{ \Nil, \Success } \neq \set{\Success}$.
  We used \ACCS for brevity, but the argument is more general.\hfill\qed
\end{example}

We write $(\envTEST,\labs) \models \AxCmpl$ whenever an
LTS of tests  $\envTEST$ and a label abstraction $\labs$
satisfy the axioms in \Cref{tab:properties-functions-to-generate-tests}.
There, axioms (1) - (6)  
and (c1) - (c3) pertain to~$\testconvSym$, 
and suffice to prove that the convergence of a program along a finite
trace~$\trace$~is logically equivalent to passing~$\testconv{\trace}$.
In the rest of this section we assume that
$\lts \in \ltsmultiset{\NBLabels}$,
$\envTEST \in \agents{\NBLabels}$,
$\labs$ abstracts $(\lts, \envTEST)$ and~$(\labs, \envTEST) \models \axcmpl$.
\begin{lemma}[\coqCpltAS{must_iff_cnv}]
	\label{lem:must-iff-cnv}

        For every trace~$\trace \in \Actfin$,
        and $\server$ in~$\lts$,
        we have that $\Must{\server}{ \testconv{ \co{\trace}} }$ if and only if~$\server \cnvalong \trace$.
\end{lemma}

We now use axioms (1) - (6) and (a\ref{ta-does-no-tau}) - (a\ref{ta-transition-to-good}), \ie the properties that pertain to $\testaccSym$,
  to show how the \mustpreorder relates the interpretations
  of the processes.
\begin{lemma}[\coqCpltAS{must_ta_or_empty_pre_action_set_for_all_trace}]
	\label{lem:must-ta-or-empty-pre-action-set-for-all-trace}
        If~$\labs$ is finitary,  
        for every $\serverA$ in $\lts$, %
        trace $\trace \in \Effectfin$, %
        and $\preactionsubset \subseteq \PreA$,
        if $\serverA \cnvalong \trace$ then either
	\begin{inparaenum}[(i)]
	\item\label{pt:must-ta-or-empty-pre-action-set-for-all-trace-1}
          $\musti{\serverA}{\testacc{\co{\trace}}{\bigcup \interp{}{\serverA}(\trace) \setminus \preactionsubset }}$, or
		\item\label{pt:must-ta-or-empty-pre-action-set-for-all-trace-2}
                  there exists $\widehat{R} \in \interp{}{\serverA}(\trace)$ such that $\widehat{R} \subseteq \preactionsubset$.
	\end{inparaenum}
\end{lemma}
\begin{proof}[Proof outline]
  Either (\ref{pt:must-ta-or-empty-pre-action-set-for-all-trace-2}) holds,
  or for each $\serverD \in \ws{}{\serverA}{\trace}$ we have
  $\map{\coRsym}{\labs}{\serverD} \setminus E \neq \emptyset$.
  Note that since waiting states are finite and $\labs$ is finitary,
  the latter property is decidable.
  This means that for every $\serverD \in \ws{}{\serverA}{\trace}$
  we can choose an element of the set $\map{\coRsym}{\labs}{\serverD}
  \setminus E$ to build a set $Z \subseteq \bigcup
  \interp{\lts}{\serverA}(\trace) \setminus \preactionsubset$
  such that $\Must{\serverA}{ \testacc{\co{\trace}}{ Z } }$.
  We obtain (\ref{pt:must-ta-or-empty-pre-action-set-for-all-trace-1})
  from the following monotonicity property:
  if $\Must{\serverA}{\testacc{\trace}{\preactionsubset_1}}$
  and $\preactionsubset_1 \subseteq \preactionsubset_2$ then
  $\Must{\serverA}{\testacc{\trace}{\preactionsubset_2}}$.
\end{proof}
\noindent
The argument sketched above highlights that the hypothesis that~$\labs$
be finitary avoids the need for the axiom of choice to select the
elements to define the set~$Z$.

We state a last lemma. We have all the lemmas to prove \Cref{thm:main-result}(\ref{pt:main-result-completeness}).
\begin{lemma}[\coqCpltAS{not_must_ta_without_required_acc_set}]
	\label{lem:not-must-ta-without-required-acc-set}
        For every $\serverB, \serverC$ in $\lts$,
        every trace $\trace \in \Effectfin$,
	and every $E \subseteq X_\labs$, if
	$\serverC \in \ws{}{\serverB}{\trace}$
	then
        $\NMust{ \serverB }{ \testacc{\co{\trace}}{
            E
            \setminus
            \map{\coRsym}{\labs}{\serverC}}}$.
\end{lemma}
\begin{proposition}[Completeness, \coqCpltAS{completeness_fw}]
  \label{prop:completeness}
        If~$\labs$ is finitary and $\serverA, \serverB$ in $\lts$,
        $\serverA \testleq \serverB$ implies $\serverA \altleq[\labs] \serverB$.
\end{proposition}
\renewcommand{\unionofpreactionspars}[3]{( \bigcup_{{#1}' \in {\ws{}{#1}{#2}}} {#3}({#1}') )}
\begin{proof}
  Fix two $\serverA, \serverB$ in $\lts$
  such that $\serverA \testleq \serverB$.
  We have to prove that for every $\trace \in \Actfin$
  if $\serverA \cnvalong \trace$ then
  \begin{inparaenum}[(i)]
  \item \label{cmp-1} $ \serverB \cnvalong \trace$, and      
  \item \label{cmp-2} $  \interp{}{\serverA}( \trace )    \ll    \interp{}{\serverB}( \trace )$.
  \end{inparaenum}
  Part (\ref{cmp-1}) is a corollary of \Cref{lem:must-iff-cnv}.
  To prove (\ref{cmp-2}), pick a $\trace \in \Actfin$ such that $ \serverA \cnvalong \trace $. We have to explain why
  $
  \interp{}{\serverA}( \trace )
  \ll
  \interp{}{\serverB}( \trace )$.
  Part (\ref{cmp-1}) ensures that $\interp{}{\serverB}( \trace )$ is defined.
  Let $R \in \interp{}{\serverB}( \trace )$.
  We have to exhibit a set
  $\widehat{\somereadyset} \in \interp{}{\serverA}( \trace )$
  such that
  $\widehat{\somereadyset} \subseteq \somereadyset$.
  By definition
  $
  \somereadyset \in \interp{}{\serverB}( \trace )$
  means that for some $\serverC$ we have $\serverB \wt{ \trace } \serverC \stable$ and $\somereadyset = \map{\coRsym}{\labs}{\serverC}$.
  Let
  $
  E \eqdef
  \unionofpreactionspars{\serverA}{\trace}{\mapSym{\coRsym}{\labs}}             \setminus
  \map{\coRsym}{\labs}{\serverC}
  $.
  \Cref{lem:must-ta-or-empty-pre-action-set-for-all-trace} states that
  either
  \begin{inparaenum}[(a)]
  \item
    $\Must{\serverA}{\testacc{\trace}{E}}$, or
  \item
    there exists a $\widehat{\somereadyset} \in
    \interp{}{\serverA}( \trace )$
    such that $\widehat{\somereadyset} \subseteq \map{\coRsym}{\labs}{\serverC}$.
  \end{inparaenum}
  To conclude the	argument it suffices to prove that (a) is false.
  To do so, we apply \Cref{lem:not-must-ta-without-required-acc-set} to prove
  $\Nmusti{ \serverB  }{ \testacc{ \trace }{ E }}$, and the
  contrapositive of $\serverA \testleq \serverB$ implies
  $\Nmusti{ \serverA }{ \testacc{ \trace }{ E }}$.
\end{proof}

In view of \Cref{prop:completeness},
an obvious question is whether 
the axioms $\AxCmpl$ can be satisfied, like~$\testconvSym$
and~$\testaccSym$ actually exist (and can be defined). 
For the four calculi discussed in \Cref{sec:applications}
they do.
We present those for calculi with value-passing. 
They are defined as
$\testconv{\trace} \eqdef g( \trace ,\TauP{\Success})$, and 
  $\testacc{\trace}{\preactionsubset} \eqdef g(\trace ,h(\preactionsubset))$, where
the auxiliary function $g$ is defined by:
$$
\begin{array}{lll}
		g(\emptytrace, \test) & \eqdef & \test  \\
		g(\ActOutV{\somechannel}{\somevalue}.\trace, \test) & \eqdef &  
		\begin{cases}
			\mailbox{\ActOutV{\somechannel}{\somevalue}} \Par g( \trace , \test) & (\VACCS) \\
			\OutputV{\somechannel}{\somevalue}{g(\trace , \test)} \extc \TauP{\Success} & (\VCCS) \\
		\end{cases}
                \\
		g(\ActInV{\somechannel}{\somevalue}.\trace, \test) & \eqdef & \InputV{\somechannel}{x}{\IfTE{ x = \somevalue}{g(\trace , \test)}{\Success}} \extc \TauP{\Success}
	\end{array}
	$$
        and for \VCCS we let
        $h(\preactionsubset) \eqdef
        \Sigma\setof{ \InputV{\somechannel}{x}{\Success}
        }{\inputc{\somechannel} \in \preactionsubset } + \Sigma\setof{
          \OutputV{\somechannel}{\somevalue}{\Success}
        }{\outputc{\somechannel} \in \preactionsubset }$, while for
        \VACCS we let $h(\preactionsubset) \eqdef \Sigma\setof{ \InputV{\somechannel}{x}{\Success} }{\inputc{\somechannel} \in \preactionsubset }$.

\section{Conclusion}
\label{sec:conclusion}

Since Morris, proposal contextual preorders have been
heavily studied and \cite{DBLP:journals/jlap/AubertV22}
provides a meta-analysis of the followed approaches.
The problem has been attacked in the settings of denotational,
operational, and game semantics, see for instance
\cite{DBLP:conf/lics/HarmerM99,DBLP:journals/iandc/McCusker00,DBLP:journals/tcs/Hennessy02,DBLP:journals/cacm/Hoare83a,PittsAM:opebtp,DBLP:journals/pacmpl/Jaber20,DBLP:conf/esop/BorthelleHJZ25,Castellan2023}.
In the setting of applicative languages, characterisations 
exist for different evaluation strategies
(call-by-value \cite{DBLP:journals/entcs/Lassen99}, %
call-by-name \cite{DBLP:conf/lics/Lassen05},
call-by-push-value \cite{DBLP:conf/cpp/ForsterSSS19}),
for calculi with a state
\cite{DBLP:journals/jacm/KoutavasLT25,DBLP:journals/pacmpl/Jaber20,DBLP:conf/fossacs/BiernackiLP19,DBLP:conf/ac/Pitts00,DBLP:journals/jfp/MasonT91}
and with value passing \cite{DBLP:journals/iandc/HartonasH98}.

In the context of languages for message-passing,
the literature comprises the results obtained by Hennessy's school
starting from~\cite{DBLP:journals/tcs/NicolaH84,DBLP:books/daglib/0066919}.
Since the \mustpreorder coincides with Hoare~\cite{DBLP:journals/cacm/Hoare83a} failure-divergence refinement (FDR)~\cite{DBLP:journals/tcs/NicolaH84,DBLP:conf/esop/BernardiCLS25},
results on the latter are worthwhile for the study of contextual relations.
Note in passing that our \Cref{thm:main-result} holds also
for FDR.
The papers on {\bfseries ioco}-testing build on top of the alternative
characterisations of~\cite{DBLP:journals/tcs/NicolaH84,DBLP:books/daglib/0066919} to develop (industrial) applications (see for
instance~\cite{DBLP:journals/jlp/BeoharM16}).

To the best of our knowledge,  \Cref{thm:main-result} is the first
constructive and machine checked characterisation of a liveness
preserving contextual preorder for (a semantic model of) calculi
with (a)synchronous communications, and with(out) value-passing.
To state and prove it, we parameterised the class of \LTSs we work with
by a set of non-blocking action $\NBLabels$, and we introduced the
notion of {\em label abstraction} $\labs$ (\Cref{def:label-abstraction}).
We used $\NBLabels$ and $\labs$ to generalise the definition of
Hennessy's alternative preorder \cite{DBLP:books/daglib/0066919},
and the proof technique introduced in~\cite{DBLP:conf/esop/BernardiCLS25}.
As in \cite{DBLP:conf/esop/BernardiCLS25}, the proof relies on the
axioms of $\testaccSym$, the rule $\scom$ of the parallel composition
of $\LTSs$ and mainly on the base case.

This required finding the axioms (given in
\Cref{tab:properties-functions-to-generate-tests})
on the contexts {\em and} on label abstractions,
that are sufficient to obtain the completeness result,
\ie \Cref{prop:completeness}, in an effective manner.
In \Cref{sec:applications} we have shown how \Cref{thm:main-result}
implies at least three existing results for concrete calculi,
and provide a new one for asynchronous value-passing \CCS.


{\bfseries Related works} %
\label{sec:related-works}
  We used as model of concurrent systems the \LTSs in
  $\ltsmultiset{\NBLabels}$.
  By definition they enjoy the axioms in \Cref{fig:axioms}, and thus
  they are a true concurrency model {\em à la}
  \cite{DBLP:conf/lics/MelliesS20,DBLP:books/sp/FajstrupGHMR16,DBLP:conf/concur/MelliesM07,DBLP:conf/mfcs/Morin05,DBLP:journals/tcs/BrachoDK97,DBLP:conf/mfps/Stark89}. 
  We have parameterised~$\ltsmultiset{\NBLabels}$ explicitly over the set $\NBLabels$ of non-blocking actions,
  and implicitly over blocking ones. This is a mild generalisation of the approach used in
  {\bfseries ioco}-testing~\cite{DBLP:conf/facs2/CuyckAT24,DBLP:journals/jlp/BeoharM16},
  where \LTSs are usually parameterised over countable sets of input labels and output labels.

  Characterisations of the \mustpreorder for value-passing calculi,
  along with their applications, have been studied in
  \cite{DBLP:journals/iandc/HennessyI93,DBLP:conf/concur/CleavelandR94,DBLP:journals/tcs/NicolaP00}.
  The calculus presented in \cite{DBLP:journals/iandc/HennessyI93} is
  an extension of \CCS without $\tau$'s, and
  we already explained how our \Cref{thm:main-result} implies the
  operational characterisation presented in that paper.
 
  The authors of \cite{DBLP:conf/concur/CleavelandR94} investigate
  the relation between abstract interpretation \cite{rival2020introduction},
  and in particular of Galois connections, and the ``testing
  preorder'', \ie the intersection of $\opMay$ and \mustpreorder
  of \cite{DBLP:journals/iandc/HennessyI93},
  which they denote~$\sqsubseteq_A$.
  In their Theorem 18 they show that under suitable hypothesis,
  any process~$p$ is improved on by its abstraction (\ie $p
  \sqsubseteq_A \alpha(p)$, where $\alpha$ is the abstraction function).
  Since our development depends on label abstractions, it is
  natural to wonder whether we are already actually working with any Galois
  connections,  and how to compose Galois connections with label
  abstractions.
  We leave this as open question.

  The authors \cite{DBLP:journals/tcs/NicolaP00} define a process
  algebra with both blocking and non-blocking actions. Their proof
  of completeness uses tests that have only blocking actions, thereby
  showing that synchronous tests, if the they can be expressed,
  suffice also in an asynchronous settings.
  A direct comparison between out \Cref{thm:main-result} and their
  Theorem 5.13 is hindered by the use of patterns in labels of the
  LTS of \cite{DBLP:journals/tcs/NicolaP00}. 

{\bfseries Future works}
Our \Cref{thm:main-result} is a necessary and intermediate step to
characterise the \mustpreorder for QuantumCCS~\cite{CeragioliGLT24}.
Efforts in this direction have started~\cite{DBLP:conf/isola/CeragioliGLT24}, but the problem remains open~\cite{MousaviPS24}.
Characterising the \mustpreorder remains an open problem also
in semantic models of non-determinism and general probabilistic choice \cite{DBLP:conf/lics/MioSV21}; in models where the communication
medium are shared stacks
\cite{PeytonJones2007BeautifulConcurrency,DBLP:conf/podc/HerlihyLMS03,HarrisMarlowPeytonJonesHerlihy2005},
shared queues
\cite{DBLP:journals/jacm/BrandZ83,DBLP:conf/lata/DengZDZ13,DBLP:journals/computers/BereczkyHT24},
or shared memory
\cite{DBLP:conf/ac/BentonHM00,metayer2004statemonadsalgebras}.
In the setting of programming languages with dynamic binding like the
$\pi$-calculus
\cite{DBLP:journals/tcs/Hennessy02,DBLP:journals/fac/DengT12} a {\em
  constructive} characterisation is still lacking.
Our \Cref{thm:main-result}, along with its mechanisation,
is a step towards these results.
Arriving at them may require generalising the notion of
duality (as in session types~\cite{DBLP:conf/tgc/BernardiDGK14,DBLP:journals/mscs/BernardiH16,DBLP:journals/corr/abs-2004-01322}), and using not only co-ready sets but also some notion of co-trace.

  Both \cite{DBLP:journals/iandc/HennessyI93,DBLP:journals/tcs/NicolaP00}
  present denotational and axiomatic characterisations of the
  \mustpreorder.
  We leave developing analogous results as future works, and here we remark
  that choice trees~\cite{DBLP:journals/pacmpl/ChappeHHZZ23} seem an obvious
  candidate to port into dependent type theory the denotational
  results obtained via Hennessy's acceptance trees~\cite{DBLP:journals/jacm/Hennessy85}.

  In synchronous settings, deciding~$\testeq[\envTEST]$ is a PSPACE
  problem \cite{DBLP:journals/iandc/KanellakisS90}.
  \Cref{cor:characterisation-for-all-CCS} makes us 
  wonder to what degree the techniques of
  \cite{DBLP:journals/iandc/KanellakisS90,DBLP:conf/aplas/BonchiCPS13}
  can be adapted to asynchronous semantics.



%
\clearpage
\bibliographystyle{plainurl}%
\bibliography{bibliography.bib}%

\begin{thebibliography}{10}

\bibitem{DBLP:journals/tcs/Abramsky87}
Samson Abramsky.
\newblock Observation equivalence as a testing equivalence.
\newblock {\em TCS}, 1987.

\bibitem{DBLP:journals/jacm/AcetoH92}
Luca Aceto and Matthew Hennessy.
\newblock {Termination, Deadlock, and Divergence}.
\newblock {\em J. {ACM}}, 1992.

\bibitem{DBLP:conf/fossacs/AcetoI99}
Luca Aceto and Anna Ing{\'{o}}lfsd{\'{o}}ttir.
\newblock {Testing Hennessy-Milner Logic with Recursion}.
\newblock In {\em FOSSACS}, 1999.

\bibitem{DBLP:journals/tcs/AmadioCS98}
Roberto~M. Amadio, Ilaria Castellani, and Davide Sangiorgi.
\newblock On bisimulations for the asynchronous pi-calculus.
\newblock {\em TCS}, 1998.

\bibitem{DBLP:journals/jlap/AubertV22}
Cl{\'{e}}ment Aubert and Daniele Varacca.
\newblock Processes against tests: On defining contextual equivalences.
\newblock {\em J. LAMP}, 2022.

\bibitem{DBLP:conf/birthday/BaldanBGV15}
Paolo Baldan, Filippo Bonchi, Fabio Gadducci, and Giacoma~Valentina Monreale.
\newblock {Asynchronous Traces and Open Petri Nets}.
\newblock In {\em Programming Languages with Applications to Biology and
  Security}, 2015.

\bibitem{DBLP:conf/ac/BentonHM00}
Nick Benton, John Hughes, and Eugenio Moggi.
\newblock Monads and effects.
\newblock In {\em APPSEM}.

\bibitem{DBLP:journals/jlp/BeoharM16}
Harsh Beohar and Mohammad~Reza Mousavi.
\newblock Input-output conformance testing for software product lines.
\newblock {\em J. LAMP}, 2016.

\bibitem{DBLP:journals/computers/BereczkyHT24}
P{\'{e}}ter Bereczky, D{\'{a}}niel Horp{\'{a}}csi, and Simon~J. Thompson.
\newblock Program equivalence in the erlang actor model.
\newblock {\em Comput.}, 2024.

\bibitem{DBLP:conf/tgc/BernardiDGK14}
G.~Bernardi, O.~Dardha, S.~J. Gay, and D.~Kouzapas.
\newblock {On Duality Relations for Session Types}.
\newblock In {\em TGC}, 2014.

\bibitem{DBLP:conf/esop/BernardiCLS25}
Giovanni Bernardi, Ilaria Castellani, Paul Laforgue, and L{\'{e}}o Stefanesco.
\newblock Constructive characterisations of the must-preorder for asynchrony.
\newblock In {\em ESOP}, 2025.

\bibitem{DBLP:journals/corr/BernardiH15}
Giovanni Bernardi and Matthew Hennessy.
\newblock {Mutually Testing Processes}.
\newblock {\em LMCS}, 2015.

\bibitem{DBLP:journals/corr/BernardiH13}
Giovanni Bernardi and Matthew Hennessy.
\newblock Using higher-order contracts to model session types.
\newblock {\em LMCS}, 2016.

\bibitem{DBLP:journals/scp/BernardiF18}
Giovanni~Tito Bernardi and Adrian Francalanza.
\newblock Full-abstraction for client testing preorders.
\newblock {\em Sci. Comput. Program.}, 2018.

\bibitem{DBLP:journals/mscs/BernardiH16}
Giovanni~Tito Bernardi and Matthew Hennessy.
\newblock Modelling session types using contracts.
\newblock {\em MSCS}, 2016.

\bibitem{DBLP:conf/fossacs/BiernackiLP19}
Dariusz Biernacki, Sergue{\"{\i}} Lenglet, and Piotr Polesiuk.
\newblock A complete normal-form bisimilarity for state.
\newblock In {\em FOSSACS}, 2019.

\bibitem{DBLP:conf/aplas/BonchiCPS13}
Filippo Bonchi, Georgiana Caltais, Damien Pous, and Alexandra Silva.
\newblock {Brzozowski's and Up-To Algorithms for Must Testing}.
\newblock In {\em APLAS}, 2013.

\bibitem{DBLP:journals/iandc/BorealeNP02}
Michele Boreale, Rocco~De Nicola, and Rosario Pugliese.
\newblock {Trace and Testing Equivalence on Asynchronous Processes}.
\newblock {\em Inf. Comput.}, 2002.

\bibitem{DBLP:conf/esop/BorthelleHJZ25}
Peio Borthelle, Tom Hirschowitz, Guilhem Jaber, and Yannick Zakowski.
\newblock An abstract, certified account of operational game semantics.
\newblock In {\em ESOP}, 2025.

\bibitem{boudol:inria-00076939}
G{\'{e}}rard Boudol.
\newblock {Asynchrony and the Pi-calculus}.
\newblock Research Report RR-1702, {INRIA}, 1992.

\bibitem{DBLP:journals/tcs/BrachoDK97}
Felipe Bracho, Manfred Droste, and Dietrich Kuske.
\newblock Representation of computations in concurrent automata by dependence
  orders.
\newblock {\em TCS}, 1997.

\bibitem{DBLP:journals/jacm/BrandZ83}
Daniel Brand and Pitro Zafiropulo.
\newblock On communicating finite-state machines.
\newblock {\em J. {ACM}}, 1983.

\bibitem{Castellan2023}
Simon Castellan, Pierre Clairambault, and Glynn Winskel.
\newblock {The Mays and Musts of Concurrent Strategies}.
\newblock In {\em {Samson Abramsky on Logic and Structure in Computer Science
  and Beyond}}, 2023.

\bibitem{DBLP:conf/fsttcs/CastellaniH98}
Ilaria Castellani and Matthew Hennessy.
\newblock {Testing Theories for Asynchronous Languages}.
\newblock In {\em FSTTCS}, 1998.

\bibitem{CeragioliGLT24}
Lorenzo Ceragioli, Fabio Gadducci, Giuseppe Lomurno, and Gabriele Tedeschi.
\newblock Quantum bisimilarity via barbs and contexts: Curbing the power of
  non-deterministic observers.
\newblock (POPL), 2024.

\bibitem{DBLP:conf/isola/CeragioliGLT24}
Lorenzo Ceragioli, Fabio Gadducci, Giuseppe Lomurno, and Gabriele Tedeschi.
\newblock Testing quantum processes.
\newblock In {\em ISoLA}, 2024.

\bibitem{DBLP:journals/pacmpl/ChappeHHZZ23}
Nicolas Chappe, Paul He, Ludovic Henrio, Yannick Zakowski, and Steve Zdancewic.
\newblock Choice trees: Representing nondeterministic, recursive, and impure
  programs in coq.
\newblock {\em POPL}, 2023.
\newblock \href {https://doi.org/10.1145/3571254} {\path{doi:10.1145/3571254}}.

\bibitem{DBLP:journals/iandc/CleavelandDSY99}
R.~Cleaveland, Z.~Dayar, S.~A. Smolka, and S.~Yuen.
\newblock {Testing Preorders for Probabilistic Processes}.
\newblock {\em Inf. Comput.}, 1999.

\bibitem{DBLP:conf/avmfss/CleavelandH89}
Rance Cleaveland and Matthew Hennessy.
\newblock {Testing Equivalence as a Bisimulation Equivalence}.
\newblock In Joseph Sifakis, editor, {\em Automatic Verification Methods for
  Finite State Systems}, 1989.

\bibitem{DBLP:conf/concur/CleavelandR94}
Rance Cleaveland and James Riely.
\newblock Testing-based abstractions for value-passing systems.
\newblock In {\em CONCUR}, 1994.

\bibitem{DBLP:journals/tcs/NicolaP00}
Rocco {De~Nicola} and Rosario Pugliese.
\newblock Linda-based applicative and imperative process algebras.
\newblock {\em TCS}, 2000.

\bibitem{DBLP:conf/lata/DengZDZ13}
Xiaojie Deng, Yu~Zhang, Yuxin Deng, and Farong Zhong.
\newblock The buffered {\(\pi\)}-calculus: {A} model for concurrent languages.
\newblock In {\em LATA}, 2013.

\bibitem{DBLP:journals/lmcs/DengGHM08}
Y.~Deng, R.~J. van Glabbeek, M.~Hennessy, and C.~Morgan.
\newblock {Characterising Testing Preorders for Finite Probabilistic
  Processes}.
\newblock {\em Log. Methods Comput. Sci.}, 2008.

\bibitem{DBLP:journals/fac/DengT12}
Yuxin Deng and Alwen Tiu.
\newblock Characterisations of testing preorders for a finite probabilistic
  {\(\pi\)}-calculus.
\newblock {\em Formal Aspects Comput.}, 2012.

\bibitem{DBLP:conf/esop/DengGMZ07}
Yuxin Deng, Rob~J. van Glabbeek, Carroll Morgan, and Chenyi Zhang.
\newblock Scalar outcomes suffice for finitary probabilistic testing.
\newblock In {\em ESOP}, 2007.

\bibitem{DBLP:books/sp/FajstrupGHMR16}
Lisbeth Fajstrup, Eric Goubault, Emmanuel Haucourt, Samuel Mimram, and Martin
  Raussen.
\newblock {\em Directed Algebraic Topology and Concurrency}.
\newblock 2016.

\bibitem{DBLP:series/txtcs/Fokkink07}
Wan~J. Fokkink.
\newblock {\em Modelling Distributed Systems}.
\newblock Texts in Theoretical Computer Science. An {EATCS} Series. 2007.

\bibitem{DBLP:conf/cpp/ForsterSSS19}
Yannick Forster, Steven Sch{\"{a}}fer, Simon Spies, and Kathrin Stark.
\newblock Call-by-push-value in coq: operational, equational, and denotational
  theory.
\newblock In {\em CPP}, 2019.

\bibitem{DBLP:journals/corr/abs-2004-01322}
Simon~J. Gay, Peter Thiemann, and Vasco~T. Vasconcelos.
\newblock Duality of session types: The final cut.
\newblock In {\em PLACES}, {EPTCS}, 2020.

\bibitem{DBLP:conf/lics/HarmerM99}
Russell Harmer and Guy McCusker.
\newblock A fully abstract game semantics for finite nondeterminism.
\newblock In {\em LICS}, 1999.

\bibitem{HarrisMarlowPeytonJonesHerlihy2005}
Tim Harris, Simon Marlow, Simon~Peyton Jones, and Maurice Herlihy.
\newblock Composable memory transactions.
\newblock In {\em PPoPP}, 2005.

\bibitem{DBLP:journals/iandc/HartonasH98}
Chrysafis Hartonas and Matthew Hennessy.
\newblock Full abstractness for a functional/concurrent language with
  higher-order value-passing.
\newblock {\em Inf. Comput.}, 1998.

\bibitem{DBLP:conf/programm/Hennessy82}
Matthew Hennessy.
\newblock Powerdomains and nondeterministic recursive definitions.
\newblock In {\em International Symposium on Programming}, 1982.

\bibitem{DBLP:journals/jacm/Hennessy85}
Matthew Hennessy.
\newblock {Acceptance Trees}.
\newblock {\em J. {ACM}}, 1985.

\bibitem{DBLP:books/daglib/0066919}
Matthew Hennessy.
\newblock {\em Algebraic theory of processes}.
\newblock {MIT} Press series in the foundations of computing. {MIT} Press,
  1988.

\bibitem{Hennessy1993ModelPiCalculus}
Matthew Hennessy.
\newblock A model for the $\pi$-calculus.
\newblock {\em TCS}, 1993.

\bibitem{DBLP:journals/tcs/Hennessy02}
Matthew Hennessy.
\newblock A fully abstract denotational semantics for the pi-calculus.
\newblock {\em TCS}, 2002.

\bibitem{DBLP:journals/jlp/Hennessy05}
Matthew Hennessy.
\newblock The security pi-calculus and non-interference.
\newblock {\em J. LAMP}, 2005.

\bibitem{DBLP:books/daglib/0018113}
Matthew Hennessy.
\newblock {\em A distributed Pi-calculus}.
\newblock Cambridge University Press, 2007.

\bibitem{DBLP:journals/iandc/HennessyI93}
Matthew Hennessy and Anna Ing{\'{o}}lfsd{\'{o}}ttir.
\newblock A theory of communicating processes with value passing.
\newblock {\em Inf. Comput.}, 1993.

\bibitem{DBLP:conf/podc/HerlihyLMS03}
Maurice Herlihy, Victor Luchangco, Mark Moir, and William N.~Scherer III.
\newblock Software transactional memory for dynamic-sized data structures.
\newblock In {\em PODC}, 2003.

\bibitem{DBLP:journals/cacm/Hoare83a}
C.~A.~R. Hoare.
\newblock {Communicating Sequential Processes (Reprint)}.
\newblock {\em Commun. {ACM}}, 1983.

\bibitem{DBLP:conf/ecoop/HondaT91}
Kohei Honda and Mario Tokoro.
\newblock {An Object Calculus for Asynchronous Communication}.
\newblock In Pierre America, editor, {\em ECOOP}, 1991.

\bibitem{DBLP:journals/pacmpl/Jaber20}
Guilhem Jaber.
\newblock Syteci: automating contextual equivalence for higher-order programs
  with references.
\newblock {\em POPL}, 2020.

\bibitem{PeytonJones2007BeautifulConcurrency}
Simon~Peyton Jones.
\newblock Beautiful concurrency.
\newblock In {\em Beautiful Code}. 2007.

\bibitem{DBLP:journals/iandc/KanellakisS90}
Paris~C. Kanellakis and Scott~A. Smolka.
\newblock {CCS} expressions, finite state processes, and three problems of
  equivalence.
\newblock {\em Inf. Comput.}, 1990.

\bibitem{DBLP:conf/esop/KoutavasH11}
Vasileios Koutavas and Matthew Hennessy.
\newblock {A Testing Theory for a Higher-Order Cryptographic Language -
  (Extended Abstract)}.
\newblock In {\em ESOP}, 2011.

\bibitem{DBLP:journals/jacm/KoutavasLT25}
Vasileios Koutavas, Yu{-}Yang Lin, and Nikos Tzevelekos.
\newblock Fully abstract normal form bisimulation for call-by-value {PCF}.
\newblock {\em J. {ACM}}, 2025.

\bibitem{DBLP:journals/entcs/Lassen99}
S{\o}ren~B. Lassen.
\newblock Bisimulation in untyped lambda calculus: B{\"{o}}hm trees and
  bisimulation up to context.
\newblock In {\em MFPS}, 1999.

\bibitem{DBLP:conf/lics/Lassen05}
S{\o}ren~B. Lassen.
\newblock Eager normal form bisimulation.
\newblock In {\em LICS}, 2005.

\bibitem{DBLP:journals/jfp/MasonT91}
Ian~A. Mason and Carolyn~L. Talcott.
\newblock Equivalence in functional languages with effects.
\newblock {\em J. Funct. Program.}, 1991.

\bibitem{DBLP:journals/iandc/McCusker00}
Guy McCusker.
\newblock Games and full abstraction for {FPC}.
\newblock {\em Inf. Comput.}, 2000.

\bibitem{DBLP:conf/concur/MelliesM07}
Paul{-}Andr{\'{e}} Melli{\`{e}}s and Samuel Mimram.
\newblock Asynchronous games: Innocence without alternation.
\newblock In {\em CONCUR}, 2007.

\bibitem{DBLP:conf/lics/MelliesS20}
Paul{-}Andr{\'{e}} Melli{\`{e}}s and L{\'{e}}o Stefanesco.
\newblock Concurrent separation logic meets template games.
\newblock In {\em LICS}, 2020.

\bibitem{metayer2004statemonadsalgebras}
François Metayer.
\newblock State monads and their algebras.
\newblock arXiv, CoRR, 2004.

\bibitem{DBLP:conf/lics/MioSV21}
Matteo Mio, Ralph Sarkis, and Valeria Vignudelli.
\newblock Combining nondeterminism, probability, and termination: Equational
  and metric reasoning.
\newblock In {\em LICS}, 2021.

\bibitem{DBLP:conf/mfcs/Morin05}
R{\'{e}}mi Morin.
\newblock Concurrent automata vs. asynchronous systems.
\newblock In {\em MFCS}, 2005.

\bibitem{morrisphd}
James~H. Morris.
\newblock {\em Lambda-calculus models of programming languages}.
\newblock PhD thesis, Massachusetts Institute of Technology, 1969.

\bibitem{MousaviPS24}
Mohammad~Reza Mousavi, Kirstin Peters, and Anna Schmitt.
\newblock Towards a formal testing theory for quantum processes.
\newblock In {\em REoCAS Colloquium in Honor of Rocco De Nicola}, 2024.

\bibitem{DBLP:journals/tcs/NicolaH84}
Rocco~De Nicola and Matthew Hennessy.
\newblock Testing equivalences for processes.
\newblock {\em Theor. Comput. Sci.}, 1984.

\bibitem{NunezLlana2008}
Manuel Núñez and Luis Llana.
\newblock A hierarchy of equivalences for probabilistic processes.
\newblock {\em FORTE}, 2008.

\bibitem{PittsAM:opebtp}
A.~M. Pitts.
\newblock Operationally-based theories of program equivalence.
\newblock In {\em Semantics and Logics of Computation}. 1997.

\bibitem{DBLP:conf/ac/Pitts00}
Andrew~M. Pitts.
\newblock Operational semantics and program equivalence.
\newblock In {\em APPSEM}, 2000.

\bibitem{DBLP:journals/iandc/RensinkV07}
Arend Rensink and Walter Vogler.
\newblock Fair testing.
\newblock {\em Inf. Comput.}, 2007.

\bibitem{rival2020introduction}
Xavier Rival and Kwangkeun Yi.
\newblock {\em Introduction to Static Analysis: An Abstract Interpretation
  Perspective}.
\newblock The MIT Press, 2020.

\bibitem{DBLP:books/daglib/0004377}
Davide Sangiorgi and David Walker.
\newblock {\em {The Pi-Calculus - a Theory of Mobile Processes}}.
\newblock 2001.

\bibitem{DBLP:conf/concur/Selinger97}
Peter Selinger.
\newblock {First-Order Axioms for Asynchrony}.
\newblock In {\em CONCUR}, 1997.

\bibitem{DBLP:journals/corr/abs-1302-1046}
Alexandra Silva, Filippo Bonchi, Marcello~M. Bonsangue, and Jan J. M.~M.
  Rutten.
\newblock Generalizing determinization from automata to coalgebras.
\newblock {\em Log. Methods Comput. Sci.}, 2013.

\bibitem{smyth_power_1978}
Michael~B. Smyth.
\newblock Power domains.
\newblock {\em J. Comput. Syst. Sci.}, 1978.

\bibitem{DBLP:conf/mfps/Stark89}
Eugene~W. Stark.
\newblock Connections between a concrete and an abstract model of concurrent
  systems.
\newblock In {\em MFPS}, 1989.

\bibitem{DBLP:journals/pacmpl/TimanyGSHGNB24}
Amin Timany, Simon~Oddershede Gregersen, L{\'{e}}o Stefanesco, Jonas~Kastberg
  Hinrichsen, L{\'{e}}on Gondelman, Abel Nieto, and Lars Birkedal.
\newblock Trillium: Higher-order concurrent and distributed separation logic
  for intensional refinement.
\newblock In {\em {POPL}}, 2024.

\bibitem{DBLP:conf/facs2/CuyckAT24}
Gijs van Cuyck, Lars van Arragon, and Jan Tretmans.
\newblock Testing compositionality.
\newblock In {\em FACS}, 2024.

\bibitem{DBLP:conf/birthday/Glabbeek22}
Rob van Glabbeek.
\newblock Fair must testing for {I/O} automata.
\newblock In {\em A Journey from Process Algebra via Timed Automata to Model
  Learning}, 2022.

\end{thebibliography}
\appendix%
\crefalias{section}{appendix}
\renewcommand{\mailbox}[1]{#1}

\section{Value Passing CCS}
\label{appendix:VCCS}

\begin{figure}[t]
\hrulefill
$$
\begin{array}{lcl}
  x & \in & \mathsf{Vars}
     \\[3pt]
  w & \in & \mathsf{Vars} \cup \Val
   \\[3pt]
  \be & :: = & x = y
  \\[3pt]
  \serverA,\serverB            & ::= &
  g 
  \BNFsep \serverA \Par \serverB 
  \BNFsep \ResChan{\somechannel}{\serverA}
  \BNFsep
 \IfTE{\be}{\serverA}{\serverB} 
\BNFsep \rec{\serverA}
\BNFsep \VRP
\\
g & ::= &
  \Nil
  \BNFsep \Unit
 \BNFsep
 \InputV{\somechannel}{x}{\serverA}
 \BNFsep \OutputV{\somechannel}{ w }{\serverA}
 \BNFsep \TauP{\serverA}
 \BNFsep \Choice{g}{g}
\end{array}
$$
\caption{Syntax of \VCCS. The meta-variables are $\somechannel \in \Names$.}
\label{fig:syntax-processes}
\hrulefill
\end{figure}

\begin{figure}[t]
\hrulefill
$$
\begin{array}{l@{\hskip 3pt}ll@{\hskip 3pt}ll@{\hskip 3pt}l}
\rinput
&
  \begin{prooftree}
    \justifies
    c?x.p \st{ c?v } p \subst{v}{x}
  \end{prooftree}
&
  \routput
&
  \begin{prooftree}
    \justifies
    c!v . p \st{ c ! v } p
  \end{prooftree}
\\[2em]

\rres
&
\begin{prooftree}
  p \st{ \alpha } p'
    \justifies
    (\nu c) p \st{\alpha} (\nu c) p'
    \using{\alpha \not \in \set{c!\_, c?\_}}
  \end{prooftree}
&
\unfold
&
  \begin{prooftree}
    \justifies
    \rec{p} \red p \subst{ \rec{p} }{ x }
  \end{prooftree}
&

  \\[2em]
  



\rthen
 &
\begin{prooftree}
  p \st{\alpha} p'
  \justifies
      {\IfTE{\be}{p}{q} \st{\alpha} p'}
      \using{\sem{\be} = True}
\end{prooftree}
&
\rname{Tau}
&
\begin{prooftree}
  \justifies
      \tau.p \st{\tau} p
\end{prooftree}
\\[2em]
\relse		
&
\begin{prooftree}
  q \st{\alpha} q'
  \justifies
      {\IfTE{\be}{p}{q} \st{\alpha} q'}
      \using{\sem{\be} = False}
\end{prooftree}
\\[2em]        
\extL
&
\begin{prooftree}
  p \st{  \action } p'
  \justifies
  p \extc q \st{  \action } p'
\end{prooftree}
&
\extR
&
\begin{prooftree}
  q \st{ \action } q'
  \justifies
  p \extc q \st{ \action } q'
\end{prooftree}
\\[2em]
  \parL
  &
\begin{prooftree}
{ \server \st{\alpha} \server' }
  \justifies
      {\server \Par \serverB \st{\alpha} \server' \Par \serverB}
\end{prooftree}
&
\parR&
\begin{prooftree}
  \serverB \st{\alpha} \serverB'
  \justifies
      {\server \Par \serverB \st{\alpha} \server \LTSPar \serverB'}
\end{prooftree}
\\[2em]
\com &
\begin{prooftree}
  \server \st{ \co{\aa} } \server' \quad \serverB \st{\aa} \serverB'
  \justifies
  \server \LTSPar \serverB \red \server' \LTSPar \serverB'
\end{prooftree}
\\
\end{array}
$$
\caption{The LTS of processes. 
  The meta-variables are
  $\aa, \ab \in \Act, \and \anyaction \in \Labels$.
}  
\label{fig:rules-LTS}
\hrulefill
\end{figure}

In this appendix we present the syntax and an operational
semantics of value-passing \CCS. The presentation is informed
by \cite{DBLP:journals/iandc/HennessyI93}.
We state the properties of interest for our main proofs,
then we present its asynchronous version (\VACCS),
and prove that the early-style \LTSs of both calculi
are in the class $\agents{\NBLabels}$,
given the correct choice of the parameter $\NBLabels$.

Recall from \Cref{ex:concrete-LTSs}
the countable set \Names, which is ranged
over by $a,b,c,\ldots$, and the set of values $\Val$.

The syntax of terms is given in \Cref{fig:syntax-processes}.
As usual,~$\rec{p}$ binds the variable~$\VRP$~in~$p$, and we use
standard notions of free variables, open and closed terms.

The process~$\Nil$ is the inactive one, analogous to {\tt skip} in
imperative languages, and we include the inactive process $\Unit$
merely to define syntactically the predicate $\goodSym$.
The prefix $\set{ \somechannel?x, \somechannel!w, \tau}.p $
represents a program willing to perform either
an output, in case of $\somechannel!v$,
an input, in case of $\somechannel?x$, or a $\tau$.
Note that in the term $\somechannel?x.p$ the prefix $\somechannel?x$
binds the variable $x$ in $P$.
We define in the standard manner the set of free and of
bound names of $p$, denoted respectively $\mathit{fn}(p)$ and
$\mathit{bn}(p)$.

The process $ \IfTE{\be}{p}{q}$ behaves as $p$ or~$q$
depending on the evaluation of the conditional $\be$.
The external sum $g_1 \extc g_2$ is a process that
behave as~$g_1$ or~$g_2$ depending on what
the external environment decides if both summands are guarded by
different prefix. If one of them is a $\tau$ prefix, then the external
sum can reduce to that process. If both summands are prefixed by the
same action, then the sum behaves in a non-deterministic manner.
It is
indeed a non-deterministic generalisation of the \texttt{match \ldots with}.
Parallel composition $p \Par q$ runs $p$ and $q$ concurrently,
allowing them also to interact with each other, thanks to rule \com.

Processes are {\em closed} terms, and their operational
semantics is the LTS $\LTS{\VCCS}{\Labels}{\st{}}$
defined by the rules in \Cref{fig:rules-LTS}.

The set of visible actions is \ActV, which we defined in \Cref{ex:concrete-LTSs}, and we define $\Acttau \eqdef \ActV \cup \set{ \tau }$.
The duality is also defined as in \Cref{ex:concrete-LTSs}.
Note that we have $\LTS{\modulo{\VCCS}{\equiv}}{\modulo{\st{}}{\equiv}}{\Acttau} \models \axiom{Finite-NB-Chains}$,
because recursive processes, for instance the process $\rec{a!1 \Par \VRP}$
outputs an unbounded but always finite amount of messages.

The predicate $\goodSym$ is defined following \cite{DBLP:journals/iandc/HennessyI93},
$$
\begin{array}{lll}
  \good{\Unit} \\
  \good{\ResChan{\somechannel}{\serverA}} & \mathit{if}\ \good{\serverA}\\
  \good{p \Par q}  &\mathit{if}\ \good{p}\, \mathit{or}\, \good{q}\\
  \good{p \extc q} &\mathit{if}\ \good{p}\, \mathit{or}\, \good{q}\\
  \good{\IfTE{\be}{p}{q}} &\mathit{if}\ \good{p}\, \mathit{and}\, \sem{\be} = \text{\tt true} \\
  \good{\IfTE{\be}{p}{q}} &\mathit{if}\ \good{q}\, \mathit{and}\, \sem{\be} = \text{\tt false}
\end{array}
$$



\subsection{Structural equivalence and its properties}
\label{sec:equiv}
\label{sec:structural-congruence}




\begin{figure}[t]
$$
\begin{array}{r@{\hskip 3pt}ll}

\\[1em]
\rulename{S-szero} & p \extc \Nil \equiv p \\
\rulename{S-scom} & p \extc q \equiv q \extc p \\
\rulename{S-sass} &  (p \extc q) \extc r \equiv p \extc (q \extc r)
\\[1em]
\rulename{S-pzero} & p \Par \Nil \equiv p \\
\rulename{S-pcom} & p \Par q \equiv q \Par p \\
\rulename{S-pass} &  ( p \Par q ) \Par r \equiv p \Par (q \Par r)
\\[1em]
\rulename{S-rzero} & \ResChan{\somechannel}{\Nil} \equiv \Nil\\
\rulename{S-rswap} & \ResChan{\somechannel}{\ResChan{\somechannelb}{\serverA}} \equiv \ResChan{\somechannelb}{\ResChan{\somechannel}{\serverA}}\\
\rulename{S-rscope} & \ResChan{\somechannel}{(\Prl{\serverA}{\serverB})} \equiv \Prl{\serverA}{\ResChan{\somechannel}{\serverB}} & \mathit{if} \somechannel \not \in \mathit{fn}(\serverA)
\\[1em]
\rulename{S-ifTrue} & \IfTE{\be}{\serverA}{\serverB} \equiv \serverA & \mathit{if} \sem{\be} = \True\\
\rulename{S-ifFalse} & \IfTE{\be}{\serverA}{\serverB} \equiv \serverB & \mathit{if} \sem{\be} = \False\\
\\[1em]
\rulename{S-refl}& p \equiv p \\
\rulename{S-symm} & p \equiv q & \mathit{if} \ q \equiv p\\
\rulename{S-trans} & p \equiv q & \mathit{if} \ p \equiv p' \ \mathit{and} \ p' \equiv q
\end{array}
$$
\caption{Rules to define structural congruence on \VCCS.}
\label{fig:equiv}
\hrulefill
\end{figure}

To manipulate the syntax of processes we define~$\equiv$
as the least congruence that satisfies the axioms in \Cref{fig:equiv}.
In this subsection we gather the properties of~$\equiv$
that we use in our proofs.
Since the properties of~$\equiv$
are standard we omit their proofs (which we have mechanised).


\begin{lemma}
  \label{lem:st-compatible-with-equiv}
  For every $p,q \in \VCCS$,
  $\alpha \in \Labels \wehavethat p \equiv \cdot \st{\alpha} q \implies p  \st{\alpha} \cdot \equiv q$.
\end{lemma}

\begin{corollary}
  \label{cor:equiv-preserves-transitions-modulo-equiv}
  \label{cor:equiv-preserves-transitions}
  For every $p,q \in \VCCS$, 
  $\alpha \in \Labels \wehavethat p \equiv q$ implies that
  $p \st{ \alpha } \cdot \equiv r$ if and only if $q \st{ \alpha } \cdot \equiv r$.
\end{corollary}

\begin{lemma}
  \label{lem:happy-sc}
  \label{lem:terminate-sc}
  \label{lem:acnv-sc}
  \label{lem:must-sc-client}
  \label{lem:must-sc-server}
  For every $\serverA, \serverB \in \VCCS$ if
  $\serverA \equiv \serverB$ then
  \begin{enumerate}[(1)]
  \item
    $\serverA \equiv \serverB$ and $\good{\serverA} \imply \good{\serverB}$;
  \item
    for every $\trace \in \Actfin$,
    \begin{enumerate}[(a)]
    \item
    $p \equiv q$ and $p \cnvalong \trace $ imply $q \cnvalong \trace$;
    \item
      if $\liftFW{\serverA} \cnvalong \trace$
  imply $\liftFW{\serverB} \cnvalong \trace$.
    \end{enumerate}
  \item
    if $\Must{\serverA}{\clientA}$ then $\Must{\serverB}{\clientA}$;
    
  \item
    if $\testA \equiv \testB$ and $\Must{\serverA}{\testA}$ then $\Must{\serverA}{\testB}$.
  \end{enumerate}
\end{lemma}

\begin{lemma}[(Un)zipping]
  \label{lem:zipping}
  \begin{inparaenum}[(1)]
  \item 
    \label{pt:zipping-strong}
    $\forevery \mu \in \Act \wehavethat$
    if $p \stfw{ \mu } p'$ and $q \st{\co{ \mu }} q'$ then
    $p \Par q \red  p' \Par q'$ or $p \Par q \equiv p' \Par q'$;
  \item 
    \label{pt:zipping-weak}
    $\forevery s \in \Actfin \wehavethat$
    if $p \wta{s} p'$ and $q \wt{\co{s}} q'$ then
    $ p \Par q \wt{\varepsilon} \cdot \equiv p' \Par q'$.
  \end{inparaenum}
\end{lemma}
\noindent
Since in the previous lemma we use the transition relation
enriched with forwarding, and by definition it contains the standard
transitions, the property applies also to the standard LTS.

\subsection{Asynchronous \CCS with value passing}
We now introduce the {\em asynchronous} version of $\VCCS$,
following the standard technique from process calculi.
We define the language $\VACCS$ as the set of $\VCCS$ processes
such that
\begin{enumerate}
\item in every prefix $c!v.p$ the term $p$ is $\Nil$,
\item in every external sum $ p \extc q $ both $p$ an $q$ are not an output prefix.
\end{enumerate}
This is standard, see for instance~\cite{DBLP:journals/tcs/AmadioCS98,DBLP:books/daglib/0004377,DBLP:journals/iandc/BorealeNP02}.


As shown in \Cref{tab:label-abstractions} to reason on the LTS
$\LTS{\modulo{\VACCS}{\equiv}}{\modulo{\st{}}{\equiv}}{\Act}$ we let
$\NBLabels = \OutV$.
The elements of $\NBLabels$ enjoy all the standard properties of outputs.

Recall the finite multisets $M$ from \Cref{ex:MO}, and let $ \Pi M $ denote
the parallel composition of all the output actions in $M$. This is well-defined
because the LTS of \VCCS enjoys \axiom{Finite-NB-Chains}.
Moreover let $\BLabels = \setof{ \somechannel?v }{ \somechannel \in \Names, v \in \Val }$.
\renewcommand{\stateA}{p}
\renewcommand{\stateB}{q}
\renewcommand{\state}{p}

\begin{lemma}
  \label{lem:output-shape}
  For every $\state \in \VACCS$,
  \begin{enumerate}
    \item
      for every
      $\nba \in \NBLabels$ and
      every $ \stateB $ such that
      $\stateA \st{ \nba } \stateB$ we have 
      $\stateA \equiv  \stateB \Par \mailbox{\nba}$,
    \item
      there exists an $M$ and a $q$ such that 
      $\stateA \equiv  \stateB \Par \mailbox{ \Pi M }$,
      and  $R(\stateB) \subseteq \BLabels$.
  \end{enumerate}
\end{lemma}

\begin{lemma}
  \label{lem:weak-output-swap}
  For every $\nba \in \NBLabels$ and $\alpha \in \Labels$,
  if $\state \wt{\nba.\alpha} \stateB$
  then $\state \wt{\alpha.\nba} \cdot \equiv \stateB$.
\end{lemma}
\begin{lemma}
  \label{lem:output-sets-finite}
  For every $\state \in \VACCS \wehavethat
  {\cardinality{\setof{ \nba \in \NBLabels }{\state \st{ \nba }}}} \in \N$.
\end{lemma}


\newcommand{\reducts}[3]{\setof{ \stateB \in #2 }{ #1 #3{\anyaction} \stateB } }

\begin{lemma}
  \label{lem:st-finite-image}
  \label{lem:sttau-finite-image}
  For every $\state \in \VACCS$ and $\anyaction \in \Acttau \wehavethat {\cardinality{
      \reducts{ \stateA }{\VACCS}{\st}}} \in \N$.
\end{lemma}

Thanks to \Cref{cor:equiv-preserves-transitions}, 
\Cref{lem:output-sets-finite} and
\Cref{lem:sttau-finite-image} hold also for the LTS modulo structural congruence,
i.e. \LTS{\modulo{\VACCS}{\equiv}}{\modulo{\st{}}{\equiv}}{\Acttau}.

A consequence of \Cref{lem:output-shape} is that $\LTS{\modulo{\VACCS}{\equiv}}{\modulo{\st{}}{\equiv}}{\Acttau} \in \agents{\NBLabels}$. The reason why this is true is the following lemma.
\begin{lemma}
\label{lem:ACCSmodulos-equiv-is-out-buffered-with-feedback}
$\Forevery p \in \VACCS,$ and $\nba \in \NBLabels$ the following properties are true,
\renewcommand{\descriptionlabel}[1]{
  {\hspace{\labelsep}{\textup{#1}}}}
\begin{description}
\item[\nbdelay:] 
  $\forevery \alpha \in \Labels \wehavethat p \st{\nba}\st{\alpha} p_3$ implies $ p \st{ \alpha }\st{ \nba} \cdot \equiv p_3$;
\item[\nbconfluence:]\mbox{ }\\\noindent%
  $\forevery \alpha \in \Labels \suchthat
  \alpha \not\in\set{ \tau, \nba} \wehavethat p \st{\nba} p'
  \text{ and } p \st{\alpha} p''$ imply that $p'' \st{\nba} q \text{ and }p' \st{\alpha} q$ for some $q$;
\item[\nbdeterminacy:] 
  $ p \st{\nba} p' \text{ and } p \st{\nba} p'' \imply p' \equiv p''$;
\item[\nbfeedback:] 
  $p \st{\nba} p' \st{\aa} q \implies p \red \cdot \equiv  q$;
\item[\nbtau:] 
  $ p \st{\nba} p' \text{ and } p \red p'' \imply$ that $p'
  \red q$ and $p'' \st{\nba} q$; or that $p' \red p''$;
\item[\nbdeterminacyinv:]\mbox{ }\\\noindent%
  $\forevery p'$ if there exists a $\hat{p}$ such that
   $p \st{\nba} \hat{p}$ and $p' \st{\nba} \hat{p}$
   then $p \equiv p'$.
\end{description}
\end{lemma}

Note that for trivial reasons we also have that $\LTS{\modulo{\VCCS}{\equiv}}{\modulo{\st{}}{\equiv}}{\Acttau} \in \agents{\emptyset}$.

\end{document}